\documentclass[a4paper]{article}

\usepackage[l2tabu, orthodox]{nag}
\usepackage[utf8]{inputenc}
\usepackage[pdfusetitle]{hyperref}
\usepackage{microtype}
\usepackage{subcaption}

\usepackage{mathtools}
\DeclarePairedDelimiter\paren{\lparen}{\rparen}

\usepackage{amsmath,amssymb,amsthm}
\newtheorem{proposition}{Proposition}
\newtheorem{lemma}{Lemma}
\newtheorem{theorem}{Theorem}
\newtheorem{claim}{Claim}
\newtheorem{corollary}{Corollary}
\theoremstyle{definition}
\newtheorem{definition}{Definition}
\theoremstyle{remark}
\newtheorem{example}{Example}

\newenvironment{claimproof}
{\noindent {\em Proof of claim:} }
{\hfill $\diamond$ \smallskip}

\usepackage{kbordermatrix}

\renewcommand{\kbldelim}{(}%
\renewcommand{\kbrdelim}{)}%

\usepackage{fullpage}

\usepackage{enumerate} %

\usepackage{xspace}

\newcommand{\N}{\ensuremath{\mathbb N}\xspace}

\newcommand{\F}{\ensuremath{\mathbb F}\xspace}
\newcommand{\Q}{\ensuremath{\mathbb Q}\xspace}
\newcommand{\cB}{\ensuremath{\mathcal B}\xspace}

\newcommand{\cF}{\ensuremath{\mathcal F}\xspace}
\newcommand{\cI}{\ensuremath{\mathcal I}\xspace}
\newcommand{\cP}{\ensuremath{\mathcal P}\xspace}

\newcommand{\cM}{\ensuremath{\mathcal M}\xspace}
\newcommand{\cS}{\ensuremath{\mathcal S}\xspace}
\newcommand{\cT}{\ensuremath{\mathcal T}\xspace}

\newcommand{\cY}{\ensuremath{\mathcal Y}\xspace}

\newcommand{\rep}[1]{{\widehat{#1}}}

\usepackage[usenames,dvipsnames]{color}

\usepackage{todonotes}

\DeclareMathOperator{\Pf}{Pf}
\usepackage{authblk}
\begin{document}

\title{Representative set statements for delta-matroids
  and the Mader delta-matroid}
\author{Magnus Wahlström}
\affil{Royal Holloway, University of London}

\maketitle
\begin{abstract}
  The representative sets lemma for linear matroids has many powerful 
  surprising applications in parameterized complexity, including
  improved FPT dynamic programming algorithms (Fomin et al., JACM 2016) and 
  polynomial kernelization and sparsification results for graph
  separation problems (Kratsch and Wahlström, JACM 2020).  However,
  its application can be sporadic, as it presupposes the existence of
  a linear matroid encoding a property relevant to the problem at
  hand.  Correspondingly, although its application led to several new
  kernelizations (e.g., \textsc{Almost 2-SAT} and restricted variants
  of \textsc{Multiway Cut}), there are also several problems left open
  (e.g., the general case of \textsc{Multiway Cut}). 

  We present representative sets-style statements for linear
  \emph{delta-matroids}, which are set systems that
  generalize matroids, with important connections to matching theory
  and graph embeddings. Furthermore, our proof uses a 
  new approach of \emph{sieving polynomial families}, which
  generalizes the linear algebra approach of the representative sets
  lemma to a setting of bounded-degree polynomials.
  The representative sets statements for linear delta-matroids then
  follow by analyzing the Pfaffian of the skew-symmetric matrix
  representing the delta-matroid.  Applying the
  same framework to the determinant instead of the Pfaffian recovers
  the representative sets lemma for linear matroids.  
  Altogether, this significantly extends the toolbox available for
  kernelization.
  
  As an application, we show an exact sparsification result for
  \emph{Mader networks}: Let $G=(V,E)$ be a graph and $\cT$ a
  partition of a set of terminals $T \subseteq V(G)$, $|T|=k$. A
  \emph{$\cT$-path} in $G$ is a path with endpoints in distinct parts
  of $\cT$ and internal vertices disjoint from $T$.  In polynomial
  time, we can derive a graph $G'=(V',E')$ with $T \subseteq V(G')$,
  such that for every subset $S \subseteq T$ there is a packing of
  $\cT$-paths with endpoints $S$ in $G$ if and only if there is one in
  $G'$, and $|V(G')|=O(k^3)$.  This generalizes the (undirected
  version of the) \emph{cut-covering lemma}, which corresponds to the
  case that $\cT$ contains only two blocks.
  
  To prove the Mader network sparsification result, we furthermore
  define the class of \emph{Mader delta-matroids},  and show that
  they have linear representations.  This should be of independent
  interest. 
\end{abstract}

\tableofcontents

\section{Introduction}

Kernelization, in the sense the term is used in parameterized
complexity, is a subtle and interesting notion.
Its original and most common justification is as a formalization of the notion
of \emph{efficient instance preprocessing}. 
This also agrees with an initial reading of the definition.
Informally, given an instance $I$ of a parameterized problem with
parameter $k$, a \emph{kernelization}
is a polynomial-time routine that simplifies $I$ -- for example,
by discarding irrelevant elements, or replacing parts of the
instance by smaller, functionally equivalent structures -- so that the
output $(I',k')$ is equivalent to $(I,k)$ in the sense that $(I,k)$ is
a yes-instance if and only if $(I',k')$ is, and the size $|I'|+k'$
is bounded by a function in $k$. A \emph{polynomial kernelization} is
a kernelization where the size bound is a polynomial in $k$.
Perhaps the simplest example of this kind of combinatorial kernelization is
\textsc{Vertex Cover}: Given a graph $G$ and a parameter $k$, does $G$
have a vertex cover of size at most $k$? We can kernelize
\textsc{Vertex Cover} using three simple rules. First, delete any
vertex of degree more than $k$ in $G$ (and decrease $k$ accordingly).
Second, if no such vertex exists but $G$ contains more than $k^2$
edges, reject the instance (by replacing it with a constant-size
negative output). Finally, discard any isolated vertices. The
resulting graph has $O(k^2)$ vertices and edges, hence these rules
provide a polynomial kernel for \textsc{Vertex Cover}. 
Many more examples of such combinatorial kernelization rules exist,
with a rich toolbox (LP-based reduction rules, sunflowers, graph
structural parameter decompositions, and so on);
see the book of Fomin et al.~\cite{Book_kernelization_FLSZ19} for more.

On the other hand, when studying the limits of the power of kernelization, then
the concept turns out to be closely related to abstract \emph{information representation}. 
This is particularly true for the known methods for proving
kernelization lower bounds, i.e., proving that a problem cannot be
kernelized into $O(k^{c-\varepsilon})$ bits for some constant $c$,
for any $\varepsilon > 0$, under some plausible complexity conjecture.
Informally, the methods for showing such a bound is to show than an
instance of the problem with parameter $k$ can ``encode'' the answers
of $k^c$ independent instances of an NP-hard problem, in the sense
that there is a process that takes approximately $k^c$ instances of an
NP-hard problem as input and in polynomial time produces an output
instance with parameter $k$ that is a yes-instance if and only if at
least one input instance is a yes-instance (an OR-composition). 
If such a procedure exists, then the problem has no kernel of
$O(k^{c-\varepsilon})$ bits unless the polynomial hierarchy
collapses~\cite{DellvM14Journal,Drucker15ANDcomp}.
If this holds for every constant $c \in \N$, then the problem has
no polynomial kernel at all, under the same assumption.
Much more work has been done on kernelization lower bounds, but the
above remains the core principle; see~\cite{BodlaenderJK14crosscomp,Book_kernelization_FLSZ19}.\footnote{More
  recently, some works have investigated such
  ``representation-theoretical'' aspects more directly, such as the
  recent research on \emph{CSP
    non-redundancy}~\cite{BrakensiekG25stoc,abs-2507-07942,Carbonnel22redundancy} (see also~\cite{KhannaP025stoc,KhannaPS24soda}), which is
  essentially a non-constructive notion of kernelization parameterized
  by the number of variables; and the work on \emph{boundaried
    kernelization}~\cite{AntipovK25boundaried,AntipovK25repset}.}

But interestingly, the same connection to information representation
appears also in some of the more advanced kernelization algorithms. 
Some powerful methods in kernelization rely upon \emph{encoding}
relevant information from a problem into an algebraic structure -- for
example, into systems of linear vectors or bounded-degree multivariate
polynomials --  and then using purely algebraic methods on the
encoding structure (such as arguments of bounded rank) to 
derive simplifications of the input instance.

A particularly interesting outlier here is the problem
\textsc{$T$-Cycle}: Given a graph $G=(V,E)$ and a set of vertices
$T \subseteq V$, is there a simple cycle in $G$ that visits every
vertex in $T$? This problem is FPT under parameter $k=|T|$~\cite{BjorklundHT12soda},
but more unexpectedly there is a polynomial-time \emph{compression}
of $(G,T)$: There is a process that takes $(G,k)$ as input
and in polynomial time produces an annotated matrix $A$ of $\tilde O(k^3)$ bits as output,
such that $(G,T)$ is positive if and only if $\det A$ contains a particular
term~\cite{Wahlstrom13}. Note that $|V(G)|$ and the length of the
cycle can be unbounded in terms of $k$, so this is a significant
apparent compression of information. Furthermore, the compression
works via pure algebra, not taking the combinatorial problem structure
into account at all. (We still do not know whether \textsc{$T$-Cycle}
has a proper polynomial kernel, which produces an instance of 
\textsc{$T$-Cycle} rather than an annotated matrix as output;
but the existence of the compression rules out any attempt at
excluding the existence of a polynomial kernel with any of our
existing lower-bounds methods.)

But this is an outlier. There are more examples where the algebraic
encoding is used to derive purely combinatorial consequences.
Perhaps the chief example is the \emph{cut-covering set lemma}
of Kratsch and Wahlström~\cite{KratschW20JACM}. Given a (possibly
directed) graph $G$ and a set $X$ of terminals in $G$ with total capacity $k$, 
in randomized polynomial time we can compute a set $Z \subseteq V(G)$
of $O(k^3)$ vertices such that for any bipartition $X=S \cup T$, 
$Z$ contains a minimum $(S,T)$-vertex cut in $G$. 
The key algebraic ingredients in this result are twofold. First, 
the result relies upon the existence of a class of linear matroids
known as \emph{gammoids}, where, informally, combinatorial
information about linkages in $G$ is encoded into the linear
independence of a collection of $k$-dimensional linear vectors.
Second, given a relevant linear matroid
a powerful result from matroid theory, the
\emph{representative sets} (or \emph{two families}) theorem,
allows us to draw important combinatorial consequences for the
original graph, such as \emph{irrelevant vertex} rules 
identifying vertices in the original graph that can be safely
bypassed without affecting any solution. 
The representative sets theorem is originally due to Lov\'asz~\cite{Lovasz1977},
with algorithmic improvements by Marx~\cite{Marx09-matroid} and Fomin et al.~\cite{FominLPS16JACM}.
This method gives polynomial kernels for a range of problems, including
\textsc{Odd Cycle Transversal}, \textsc{Almost 2-SAT}, and
\textsc{Multiway Cut} with $s=O(1)$ terminals, where direct combinatorial
attacks appear (so far) completely powerless~\cite{KratschW20JACM}.

Another benefit of the algebraic methods is that instead
of merely preserving a ``single bit'' of information (such as whether
an instance has a solution of size $k$ or not), as would be required
by the definition of a kernel, algebraic methods frequently preserve 
\emph{all relevant behaviours} of the system they encode, such as,
here, preserving a min-cut for every partition of $X$.

Analogous results have been achieved using methods of bounded-degree
multivariate polynomials. In particular, Jansen and Pieterse~\cite{JansenP19toct} 
studied the question of sparsification of Boolean CSPs -- i.e.,
the limits of kernelization of Boolean CSPs parameterized by the
number of variables $n$ -- and employed a highly successful method
of encoding the constraints of the instance by bounded-degree polynomials.
Here, a constraint $R(x_1,\ldots,x_r)$ is encoded by a polynomial
$p(x_1,\ldots,x_r)$ if, for every tuple $(b_1,\ldots,b_r) \in 2^r$,
$p(b_1,\ldots,b_r)=0$ if and only if the constraint accepts $(b_1,\ldots,b_r)$. Rank bounds on the
space spanned by bounded-degree polynomials then imply that the CSP
can be sparsified to a ``basis'' of $O(n^d)$ constraints, where $d$ is
the degree bound, while preserving the satisfiability status for every
assignment to the variables. For many Boolean CSPs, this yields a
polynomial kernelization of a size that is both non-trivial and
optimal up to standard conjectures; see also Chen et al.~\cite{ChenJP20}.

These two methods, of linear matroid encodings and bounded-degree
polynomials, are also similar on a technical level, as the rank bound
for a system of bounded-degree polynomials is just a rank bound on a
derived linear space. %

We consider this information-encoding aspect of kernelization to be
both one of its more interesting aspects, and a promising venue for
further improvement. In particular, there is a good number of
questions in polynomial kernelization where progress has essentially
stalled, such as the existence of a polynomial kernel for
\textsc{Directed Feedback Vertex Set} or for undirected
\textsc{Multiway Cut} with an unbounded number of terminals,
where it appears that any further progress requires significantly new
methods. 

Finally, both representative sets and polynomial methods (and, indeed,
more advanced algebraic attacks) have seen significant use in FPT and
exact algorithms~\cite{FominLPS16JACM,Bjorklund14tsp,CyganKN18hamiltonrank,Nederlof20,Brand22discr,Pratt19waring},
including the recent \emph{determinantal sieving} method that 
combines the two worlds~\cite{theoretics:14026}.
However, we focus on the kernelization aspect.

\paragraph{Our results.}
We extend the algebraic toolbox for kernelization in three ways.
First, we show how the matroid representative sets method can be
captured and extended using a method we refer to as 
\emph{bounded-degree sieving polynomials}.
This allows us to generalize the matroid manipulations in applications
of the former method by possibly more familiar arguments on
bounded-degree polynomials. It also unifies the two approaches
mentioned above.

Second, we use this method to show representative sets-like theorems
over linear \emph{delta-matroids}. Delta-matroids are a generalization
of matroids, where the notion of \emph{independence} has been replaced 
by a more general notion of \emph{feasible sets}. In particular,
where linear matroids are represented by the column space of a matrix, 
linear delta-matroids are represented by principal induced submatrices
of a skew-symmetric matrix. For example, there is a delta-matroid 
whose feasible sets are precisely the vertex sets of (not necessarily
maximum) matchings in a graph, and it is represented by the well-known
\emph{Tutte matrix}.  

Third, we show a new representable class of delta-matroid, the
\emph{Mader delta-matroid}, which generalizes the matching
delta-matroid to the setting of Mader's $\cS$-path packing theorem
(see below).
As an application, we extend the cut-covering set lemma to a setting
where the terminal set $T$ may be partitioned into more than two parts.
That is, we show the existence of an $O(|T|^3)$-vertex \emph{Mader-mimicking network},
and as a corollary, we show the existence of an $O(|T|^3)$-vertex graph 
that preserves the maximum cardinality of a vertex-disjoint path packing
for any partition of the terminal set $T$. Cut-covering sets
(and \emph{cut sparsifiers}~\cite{LeightonM10spars,KrauthgamerR20planar}) 
correspond to the case where only partitions with two blocks are allowed. 

We hope that this will inspire further follow-up work, and in
particular that some progress can be made on the hard questions of
polynomial kernelization.

For the rest of the introduction, we give a less technical overview of
these methods; the full technical details follow in the rest of the
paper. 

\paragraph{Recent work and improvements.} 
Since the publication of the extended abstract of this paper~\cite{Wahlstrom24repset},
Koana and the author~\cite{KoanaW25stacs} have investigated representations
of linear delta-matroids and related algorithms, and proposed 
\emph{contraction representations} as an alternative representation
(see Section~\ref{ssec:rep-new}).
Using this notion, some of the results in the paper could be generalized.
The majority of results and constructions are intact from~\cite{Wahlstrom24repset},
but Section~\ref{ssec:rep-new} shows a new, more general
result using contraction representations. 

The results of~\cite{KoanaW25stacs} also provide new constructions for
compositions of linear delta-matroids, such as the \emph{linear delta-sum},
which is a very powerful way of constructing new linear delta-matroids.
In addition, the method of determinantal sieving~\cite{theoretics:14026}
has been extended to delta-matroids~\cite{abs-2502-13654}.
Altogether, in our opinion, these results provide a very interesting
collection of algebraic methods that should be further investigated,
both for kernelization and algorithmic purposes. 

\subsection{Representative sets and sieving polynomials}

To illustrate the method of bounded-degree sieving polynomials, let us
first recall the matroid formulation of the representative sets lemma.
We here use a presentation with as little matroid terminology as
possible; a more standard, technically complete presentation is found
in Section~\ref{sec:l-m-appl}. 

Let $M=(V,\cI)$ be a linear matroid and let $A$ be a matrix
representing $M$; i.e., $\cI \subseteq 2^V$, $A$ has columns labelled
by elements of $V$, and a set $S \subseteq V$ is contained in $\cI$
if and only if the columns labelled by $S$ are linearly independent.
In other words, in shorthand, for $S \subseteq V$ we may say $S \in \cI$
if and only if the matrix $A[\cdot,S]$ is non-singular.

Let $k$ be the rank of $A$; w.l.o.g.\ we can assume that $A$ has $k$ rows.
Let $X, Y \subseteq V$ be independent sets with $|X|=k-q$ and $|Y|=q$ for some $q=O(1)$. 
We say that \emph{$Y$ extends $X$} if the columns $X \cup Y$ have rank $|X|+|Y|=k$
in $A$; i.e., if $X \cup Y$ is linearly independent and $X \cap Y = \emptyset$.
Now let $\cY \subseteq \binom{V}{q}$ be a collection of $q$-sets over $V$.
We say that a set $\rep{\cY} \subseteq \cY$ \emph{represents $\cY$ in $M$}
if the following holds: For every $X \subseteq V$, $|X|=k-q$,
if there is an element $Y \in \cY$ that extends $X$
then such an element is retained in $\rep{\cY}$.

\begin{theorem}[Matroid representative sets, simplified~\cite{Lovasz1977,Marx09-matroid,FominLPS16JACM}]
  \label{ithm:matroid-rep-set}
  Let $N=(V,\cI)$ be a linear matroid of rank $k$, and let $\cY$ 
  be a collection of $q$-sets in $M$. In polynomial time, we can
  compute a set $\rep{\cY} \subseteq \cY$ such that
  $|\rep{\cY}| \leq \binom{k}{q}$ and $\rep{\cY}$ represents $\cY$ in $M$. 
\end{theorem}

In particular, for the setting $q=O(1)$ the size of $\rep{\cY}$ is
polynomial in $k$. 

The power of this result lies in that the collection of sets $X$ does
not have to be supplied as input. Think of this theorem as
being applied at ``compile time'', while a set $X$ is supplied at
``run time'' -- or more accurately, the computation of a
representative set $\rep{\cY} \subseteq \cY$ of size $O(k^q)$
is a polynomial-time \emph{preprocessing} stage.

\paragraph{Property families and the setup.}
To generalize this situation, we introduce an abstraction.
Let $S$ be a ground set. By a \emph{property} of $S$, we mean simply a
set $P \subseteq S$, and a \emph{property family} $\cP \subseteq 2^S$
is a collection of properties. For example, in the above setting we
would have $S=\cY$ as the set of $q$-sets $Y$ supplied to the
preprocessing stage, and for every set $X \subseteq V$ with $|X|=k-q$ 
the property $P_X \subseteq S$ is simply ``$Y$ extends $X$''
(i.e., $P_X$ contains those sets $Y \in S$ that extend $X$).
Then, a \emph{representative set with respect to $\cP$}
is a set $\rep{S} \subseteq S$ such that for every non-empty property
$P \in \cP$ we have $P \cap \rep{S} \neq \emptyset$.

We then capture computation of a representative set in terms of
\emph{sieving polynomials} as follows. Let $S$ be the ground set. 
Let $\F$ be a field and let $\psi \colon S \to \F^r$ be a vector map.
For example, in the case of $q$-sets over linear matroids as above, for
a set $Y=\{v_1,\ldots,v_q\}$ we
simply have $\psi(Y)=(A[1,v_1],\ldots,A[k,v_q])$, i.e., the matrix $A[\cdot,Y]$
rolled out as a vector of length $r=kq$.
Then for any property $P \subseteq S$, a \emph{sieving polynomial for $P$}
is an $r$-variate polynomial $p$ over $\F$ such that
for any $x \in S$, $p(\psi(x)) \neq 0$ if and only if $x \in P$.
In other words, $p$ vanishes on $\psi(x)$ precisely when $x \notin S$.
We say that a property family $\cP \subseteq 2^S$ admits a
\emph{sieving polynomial family of degree $d$ (over $\F$ and $\psi$)}
if there exists a vector map $\psi \colon S \to \F^r$ 
such that for every property $P \in \cP$ there is
a sieving polynomial for $P$ of degree at most $d$. 

The central result is then the following.

\begin{theorem}[Representative set for sieving polynomials]
  \label{ithm:sieve-rep-set}
  Let $S$ be a set and $\psi \colon S \to \F^r$ a polynomial-time
  computable map over a field $\F$. Let $d$ be a constant. In
  polynomial time, we can compute a set $\rep{S} \subseteq S$
  with $|\rep{S}|=O(r^d)$ such that for every property family
  $\cP$ that admits a sieving polynomial family of degree at most $d$
  over $\psi$ and $\F$, $\rep{S}$ is a representative set with respect
  to $\cP$. 
\end{theorem}

This result is not difficult to prove (see Section~\ref{sec:sieving-polynomials})
but we find it remarkably powerful and informative regarding proving 
new representative set-style results. 

\paragraph{Illustration: Linear matroid case.}
As an illustration, let us use Theorem~\ref{ithm:sieve-rep-set}
to derive Theorem~\ref{ithm:matroid-rep-set} for the bounded-degree case.
As mentioned, the input is $(A,\cY)$ where $A$ is a matrix of dimension $k \times V$ and of
rank $k$ over some field $\F$, and $\cY$ is a collection of $q$-sized subsets of $V$,
where we assume $q=O(1)$. We need to prove that the property family
\[
  \cP = \bigl\{\{Y \in \cY : \text{$Y$ extends $X$}\} \mid X \in \binom{V}{k-q}\bigr\}
\]
has a sieving polynomial family of degree $q$. For this, let $X \subseteq V$
be a linearly independent set of vectors with $|X|=k-q$,
say $X=\{v_1,\ldots,v_{k-q}\}$,
and introduce a set of variables $Z=\{z_{i,j} \mid i \in [k], j \in [q]\}$.
Define the polynomial $P_X(Z)$ as
\[
P_X(Z) = \det \begin{pmatrix} A[\cdot,X] & Z \end{pmatrix} = \det
  \begin{pmatrix}
    a_{1,1} & \dots & a_{k,k-q} & z_{1,1} & \dots z_{1,q} \\
    \vdots &        & \vdots    &   \vdots& \vdots \\
    a_{k,1}& \dots  & a_{k,k-q}   & z_{k,1} & \dots z_{k,q}
  \end{pmatrix},
\]
where $a_{i,j}=A[i,v_j]$ is a constant depending on $X$
and $z_{i,j}$ is a variable for each $i$ and $j$. 

Now, $P_X(Z)$ is clearly a polynomial of degree $q$ over $Z$. 
We argue that $P_X$ is also a valid sieving polynomial for $X$
with respect to the map $\psi$ defined above, i.e., the map
$\psi \colon \cY \to \F^{kq}$ defined by
$\psi(\{y_1,\ldots,y_q\})=(A[1,y_1],\ldots,A[k,y_q])$
(under a natural order on $y_i$).
Indeed, if $Y=\{y_1,\ldots,y_q\} \subseteq V$, and
we let $z_{i,j}=A[i,y_j]$ for every $i \in [k]$, $j \in [q]$,
then this evaluation produces the matrix
$\begin{pmatrix} A[\cdot,X] & A[\cdot,Y] \end{pmatrix}$,
which is non-singular if and only if $X \cup Y$ is linearly independent
and $X \cap Y = \emptyset$; i.e., if and only if $Y$ extends $X$.
Since $P_X(\psi(Y))$ is simply the determinant of this matrix,
we get $P_X(\psi(Y)) \neq 0$ if and only if $Y$ extends $X$, as required. 

Our results on representative set-theorems for delta-matroids
follow the same approach as above, except starting from
skew-symmetric matrices, which represent linear delta-matroids,
using the Pfaffian instead of the determinant.
See Section~\ref{sec:delta-matroid-repsets}.

In addition, having the framework defined in terms of bounded-degree
sieving polynomials instead of explicitly depending on a matrix
representation opens the door to defining new applications, e.g., by
combining polynomial families from different sources and using 
other familiar sieving-operations over polynomials. 
We leave this avenue unexplored. 

\subsection{Delta-matroid representative sets} 

Let us now sketch the main new application of the method of sieving
polynomial families, by extending Theorem~\ref{ithm:matroid-rep-set}
to the setting of \emph{delta-matroids}.

Delta-matroids are a generalization of matroids with many applications
across discrete mathematics and computer science. They were originally
defined by Bouchet~\cite{Bouchet87DMone}, although similar
constructions had also occurred under other names.
See Moffatt~\cite{Moffatt19deltamatroids} for an overview.

A delta-matroid is a pair $D=(V,\cF)$ where $\cF \subseteq 2^V$
is a set of \emph{feasible sets}, subject to certain axioms.
We defer the full definition to Section~\ref{sec:delta-matroids} and
narrow our focus here on a subclass. Let $A \in \F^{n \times n}$ be a
square matrix. $A$ is \emph{skew-symmetric} if $A^T=-A$. 
Analogous to linear matroids being represented by the column
space of a matrix, a \emph{linear delta-matroid} is a delta-matroid
represented through a skew-symmetric matrix. 
As a subclass of linear delta-matroids, we define a
\emph{directly represented} delta-matroid to be a delta-matroid $D=(V,\cF)$
represented by a skew-symmetric matrix $A \in \F^{V \times V}$
such that for any $S \subseteq V$ we have $S \in \cF$ if and only if
$A[S]$ is non-singular. (Here, $A[S]=A[S,S]$ is the principal
submatrix induced by $S$.)

Let $A \in \F^{V \times V}$ be a skew-symmetric matrix. The
\emph{support graph} of $A$ is the graph $G=(V,E)$ where
$uv \in E$ for $u, v \in V$ if and only if $A[u,v] \neq 0$; note that $G$ is an
undirected graph. We note two algebraic properties of skew-symmetric matrices.
First, the determinant of $A$ is the square of the \emph{Pfaffian} of
$A$, an efficiently computable polynomial whose terms range over
perfect matchings of the support graph of $A$.
Our delta-matroid representative set statements follow
by extracting bounded-degree sieving polynomials from the Pfaffian,
similarly to the sketch given above for the matroid representative set
statement. Second, skew-symmetric matrices allow an algebraic
\emph{pivoting} operation, which is an important tool in this paper.
Pivoting provides the following operation: Given a delta-matroid $D=(V,\cF)$
directly represented by a matrix $A$, and a set $S \subseteq V$ that is feasible
in $D$, we can in polynomial time construct a matrix $A'=A*S$ such
that the delta-matroid $D'=(V,\cF')$ directly represented by $A'$
has the following property: For any set $T \subseteq V$,
$T$ is feasible in $D'$ if and only if $S \Delta T$ is feasible in $D$
(where $\Delta$ denotes symmetric difference). 

As a key example of a delta-matroid, let $G=(V,E)$ be a graph. Recall
that the \emph{Tutte matrix} for $G$ is a skew-symmetric matrix $A$
with rows and columns labelled by $V$, which is non-singular if and
only if $G$ has a perfect matching. By symmetry of the construction,
it is straight-forward to observe a stronger property of $A$: For
every $S \subseteq V$, $A[S]$ is non-singular if and only if $G[S]$
has a perfect matching. Thus, $A$ directly represents a linear
delta-matroid $D=(V,\cF)$ where a set $S$ is feasible if and only if
it is the endpoints of a matching in $G$.

We also provide a new extension of this construction to delta-matroids
related to Mader's path-packing theorem; see below.

\paragraph{Delta-matroid representative sets.} We now present the
representative sets statement for directly represented delta-matroids. 
For technical reasons we need to make it a little bit more involved
than the corresponding statement for linear matroids to retain its power.
Let $A \in \F^{V \times V}$ be a skew-symmetric matrix, directly
representing a delta-matroid $D=(V,\cF)$, and let $T \subseteq V$
be a set of terminals. The \emph{rank} of $T$ is the maximum
intersection $F \cap T$ for any feasible set $F$ of $D$ (where $F$ can
be assumed to be simply a basis of the column space of $A$; see
Section~\ref{sec:delta-matroid-repsets} for more precise definitions).
For a set $X \subseteq T$ and $Y \subseteq V$, we say that
\emph{$Y$ extends $X$ in $D$} if $X \cap Y = \emptyset$
and $X \cup Y$ is feasible in $D$. 
Given a collection $S \subseteq \binom{V}{q}$ of $q$-sets in $D$,
as with matroids, we say that $\rep{S} \subseteq S$
\emph{represents $S$ in $(D,T)$} if, for every $X \subseteq T$
such that there exists $Y \in S$ such that $Y$ extends $X$ in $D$, 
there exists $Y \in \rep{S}$ such that $Y$ extends $X$ in $D$.

\begin{theorem}
  \label{ithm:dm-rep-set}
  Let $D=(V,\cF)$ be a directly represented delta-matroid.
  Let $T \subseteq V$ and $S \subseteq \binom{V}{q}$ be given,
  where the rank of $T$ is $k$ and $q=O(1)$. 
  In polynomial time, we can compute a set $\rep{S} \subseteq S$
  with $|\rep{S}|=O(k^q)$ such that $\rep{S}$ represents $S$ in
  $(D,T)$. 
\end{theorem}

Note that the corresponding statement for linear matroids is the case
that $T=V$ and the rank of the whole matroid is $k$. In matroids, this
is not much of a restriction since matroids support a \emph{truncation}
operation that limits their rank (and which can be performed
efficiently over linear matroids). In delta-matroids, no such
operation is possible, and the addition of a low-rank terminal set $T$
adds significant expressive power to the statement.
See Section~\ref{sec:rep-lb} for more.

The proof of Theorem~\ref{ithm:dm-rep-set} follows from two simple
steps. First, we note that there is a pivoting operation on $D$ such
that the resulting matrix $A'=A*S$ (where $A$ represents $D$ and
$S \subseteq T$) has significant structural sparsity properties.
Second, given such a structurally sparse matrix $A'$, we prove
Theorem~\ref{ithm:dm-rep-set} by a bounded-degree sieving polynomial
argument over the Pfaffian of $A'$. 

\paragraph{Mader delta-matroids.}
To support applications of Theorem~\ref{ithm:dm-rep-set}, we define a
delta-matroid based on Mader's $\cS$-path-packing theorem, and show that
it is directly representable. Let $G=(V,E)$ be an undirected graph and
$T \subseteq E$ a set of terminals. Let $\cS$ be a partition of $T$.
An \emph{$\cS$-path} is a path $P$ in $G$ whose endpoints lie in
different blocks of $\cS$ and which is internally disjoint from $T$.
An \emph{$\cS$-path packing} is a set $\cP$ of pairwise
vertex-disjoint $\cS$-paths. By a theorem of Mader~\cite{Mader78Hpath}
(see Schrijver~\cite{SchrijverBook} and~\cite{SeboS04pathpacking,Pap06Madermatroids}), a
maximum-cardinality $\cS$-path packing can be efficiently computed,
and is characterized by a min-max theorem generalizing the Tutte-Berge
formula for matching.
An $\cS$-path packing is \emph{perfect} if it has cardinality $|T|/2$. 

Graph matchings correspond to the case that $T=V$
and every block in $\cS$ has cardinality one, and at the other end, the
case that $|\cS|=2$ corresponds to undirected vertex-disjoint max-flow.

We refer to $\cS$-paths as \emph{Mader paths}, and to $\cS$-path
packings as \emph{Mader path-packings}, for clarity and in order not
to have the meaning of the term depend on a typographical detail.
We show the following.

\begin{theorem}
  \label{ithm:mader-delta-matroid}
  Let $G=(V,E)$ be an undirected graph, $T \subseteq V$ a set of
  terminals and $\cS$ a partition of $T$. Let $D=(T,\cF)$ with
  $\cF \subseteq 2^T$ be defined so that for $S \subseteq T$,
  $S \in \cF$ if and only if there is a Mader path-packing in $(G,\cS)$
  whose endpoints is precisely the set $S$. Then $D$ is a
  delta-matroid, and furthermore a direct representation of $D$ can
  be computed in randomized polynomial time. 
\end{theorem}

We refer to $D$ defined by $G$ and $\cS$ as a \emph{Mader delta-matroid}. 
Taking $\cB$ as the set of all maximum-cardinality feasible sets in $D$
defines a matroid, known as a \emph{Mader matroid}. Mader matroids are
representable~\cite{SchrijverBook}, and in fact it has been shown that
every Mader matroid is an example of a \emph{gammoid}, i.e.,
that it can be equivalently represented in terms of linkages in
directed graphs~\cite{Pap06Madermatroids}. However, the structure of
Mader delta-matroids appears more intricate, and we are not aware of
any similar way to capture them by graph flows or matchings.\footnote{As an
  excursion, consider the analogous situation with
  \emph{matching matroids} and \emph{matching delta-matroids}.
  The set of endpoints of all \emph{maximum-cardinality} matchings
  forms a matching matroid, and has a simple combinatorial
  description using the Tutte-Berge formula. However, this is frequently
  not very informative, e.g., if the graph has a perfect matching then
  it is just the free matroid on the vertex set. 
  The matching delta-matroid, whose feasible sets are the endpoints of matchings in $G$,
  is a more complex object which captures a lot more of the structure
  of $G$; for example, the feasible sets of cardinality 2 are
  precisely the edges of $G$. Similarly, the Mader delta-matroid
  retains more information about $(G,\cS)$ than the Mader matroid does.
}

While Mader matroids are well established in the literature (see
Schrijver~\cite{SchrijverBook}), we have not been able to locate any
previous mention of Mader delta-matroids in the literature. However, the
slides from a presentation by Gyula Pap~\cite{PapSlides} mentions the
name and asks whether they are representable (to which we therefore
present a positive answer). 

\subsection{Mader-mimicking networks}

As the final contribution of the paper, we present a kernelization
application of delta-matroid representative sets. As mentioned,
one of the chief results of the kernelization applications of matroid
representative sets is of an \emph{exact cut sparsifier},
a.k.a.\ \emph{mimicking network} for terminal cuts.
More properly, let $G=(V,E)$ be a graph (possibly directed)
and let $T \subseteq V$ be a set of terminals, $|T|=k$.
A \emph{mimicking network} for $(G,T)$ is a graph $H=(V_H,E_H)$
with $T \subseteq V_H$ such that for every bipartition
$T=A \cup B$, the size of an $(A,B)$-mincut is the same in $H$ and in $G$.
Consider the special case that every vertex of $T$ has capacity~1
(i.e., the case of \emph{deletable terminals}). Using matroid
representative sets over gammoids, Kratsch and Wahlstr\"om were able to show
that in randomized polynomial time we can compute
a mimicking network $H=(V_H,E_H)$ for $(G,T)$ with $|V_H|=O(k^3)$.
Furthermore, $H$ can be constructed by computing a set $Z \subseteq V(G)$
with $|Z|=O(k^3)$ and projecting $G$ down to the vertex set $T \cup Z$. 

We show a similar statement for the more general Mader path packing
problem. Let $G=(V,E)$ be a graph, $T \subseteq V$ a set of terminals
and $\cS$ a partition of $T$. 
Say that a set $S \subseteq T$ is \emph{Mader matchable} 
if there exists a Mader path packing $\cP$ over $\cS$ such that
$V(\cP) \cap T = S$. We show that given $(G,\cS)$ with $|T|=k$
there is a $O(k^3)$-vertex \emph{Mader-mimicking network} in the
following sense. 

\begin{theorem}
  \label{ithm:mader-mimicking-network}
  Let an undirected graph $G=(V,E)$, a set $T \subseteq V$ and a
  partition $\cS$ of $T$ be given. Let $|T|=k$. In randomized
  polynomial time, we can compute a graph $H=(V_H,E_H)$ with
  $T \subseteq V_H$ such that for every $S \subseteq T$,
  $S$ is Mader matchable in $(G,\cS)$ if and only if it is Mader
  matchable in $(H,\cS)$.
\end{theorem}

Clearly, this also means that for every $S \subseteq T$, the maximum
cardinality of a Mader path packing $\cP$ with $V(\cP) \cap T \subseteq S$
is the same in $G$ as in $H$. Note that this covers undirected vertex
cuts for the case $|\cS|=2$ via Menger's theorem. Via a closer
inspection of Theorem~\ref{ithm:mader-mimicking-network} we 
can also show the following.

\begin{theorem}
  Let $G=(V,E)$ be an undirected graph and $T \subseteq V$ a set of
  terminal vertices, $|T|=k$. In randomized polynomial time, we can
  compute a graph $G'=(V',E')$, $T \subseteq V'$,
  such that $|V(G')|=O(k^3)$ and for every partition $\cT$ of $T$,
  the Mader path-packing numbers in $(G,\cT)$ and $(G',\cT)$
  are identical.
\end{theorem}

This is the path-packing dual to the recently studied notion of a
\emph{multicut-mimicking network}~\cite{Wahlstrom22talg},
i.e., a graph $H=(V_H,E_H)$ with $T \subseteq V_H$ such that the
minimum cardinality of a multicut is the same in $G$ and $H$ for every
partition of $T$. As noted in~\cite{Wahlstrom22talg}, the
cut-covering sets lemma suffices to prove a 2-approximate multicut-mimicking network 
of size $O(k^3)$ for edge multicuts (i.e., a 2-approximate multicut sparsifier). 
The above result implies the same for vertex multicuts over $T$.

\begin{corollary} \label{cor:two-approx-vmc}
  Let $G=(V,E)$ be an undirected graph and $T \subseteq V$ a set of
  terminal vertices. In polynomial time we can compute an undirected
  graph $G'=(V',E')$ such that $T \subseteq V'$, $|V'|=O(|T|^3)$
  and for any collection of cut requests $R \subseteq \binom{T}{2}$
  the size of a minimum multicut with respect to $R$ differs by
  at most a factor 2 in $G$ and $G'$. 
\end{corollary}

\paragraph{Structure of the paper.} 
Section~\ref{sec:prel} reviews further preliminaries.
Section~\ref{sec:sieving-polynomials} provides basic results on
sieving polynomial families.
Section~\ref{sec:delta-matroid-repsets} gives representative sets
statements for delta-matroids and matching lower bounds.
Section~\ref{sec:mader-delta} derives the representation for Mader
delta-matroids. 
Section~\ref{sec:mader-mimic} derives the results on Mader-mimicking
networks. 
Finally, Section~\ref{sec:conc} concludes the paper and reviews some
open questions. 

\section{Preliminaries}
\label{sec:prel}

\subsection{Matroids}

We review standard definition and notation on matroids. For more, see
Oxley~\cite{OxleyBook2}. 

A \emph{matroid} is a pair~$M=(U,\cI)$, where~$U$ is the \emph{ground set} 
and~$\cI \subseteq 2^U$ a collection of subsets of $U$, 
called \emph{independent sets}, subject to the following axioms:
\begin{enumerate}
\item $\emptyset \in \cI$
\item If $I_1 \subseteq I_2$ and $I_2 \in \cI$ then $I_1 \in \cI$
\item If $I_1, I_2 \in \cI$ and $|I_2|>|I_1|$, then there
  exists an element $x \in (I_2 \setminus I_1)$ such that $I_1+x \in \cI$
\end{enumerate}
Here and in the rest of the paper, $I+x$ and $I-x$ for a set $I$ and element $x$ 
denote $I \cup \{x\}$ and $I \setminus \{x\}$, respectively.
A set~$I \subseteq U$ is \emph{independent} if~$I \in \cI$. A
set~$B \in \cI$ is a \emph{basis} of~$M$ if it is a maximal
independent set. Note that every basis has the same cardinality.  For
a set~$X \subseteq U$, the \emph{rank}~$r(X)$ of~$X$ is the largest
cardinality of an independent set~$I \subseteq X$.
The rank of~$M$ is~$r(M):=r(U)$.
We say that a set $S \subseteq U$ \emph{spans} $M$ if $r(S)=r(M)$. 
We will be concerned with the class of \emph{linear} matroids, defined
as follows.

\begin{definition} 
  Let~$A$ be a matrix over a field~$\F$ and let~$U$ be the set of
  columns of~$A$. Let~$\cI$ be the set of all sets~$X \subseteq U$
  such that the set of column vectors indexed by $X$ is linearly
  independent over $\F$. Then~$(U,\cI)$ defines a matroid~$M$, and the
  matrix~$A$ \emph{represents}~$M$. A matroid is \emph{representable}
  (over a field $\F$) if there is a matrix (over $\F$) that
  represents it. A matroid representable over some field is called
  \emph{linear}. 
\end{definition}

It is well known that row operations on a matrix $A$ do not affect the
sets of columns of $A$ that are linearly independent. Thus, any
representation $A$ for a matroid $M$ of rank $r=r(M)$ can be reduced
to one where $A$ has only $r$ rows.
For a matroid $M=(U,\cI)$, the \emph{dual} of $M$ is the matroid $M^*=(U,\cI^*)$
where as set $S \subseteq U$ is independent if and only if $U \setminus S$ spans $M$
For a set $S \subseteq U$, the \emph{deletion} $M \setminus S$ of $S$ in $M$ 
defines the matroid on ground set $U \setminus S$, where independent
sets are sets $I \subseteq U \setminus S$ that are independent in $M$.
The \emph{contraction} $M/S$ is defined as $(M^* \setminus S)^*$.
If $S$ is independent in $M$, then the independent sets of $M/S$
are precisely the sets $T \subseteq (U \setminus S)$ such that
$S \cup T$ is independent in $M$. All these operations (dual,
deletion, contraction) can be performed efficiently on linear matroids.
For more on matroids, see Oxley~\cite{OxleyBook2}.  For more specifically on matroid
applications in parameterized complexity, see the paper of
Marx~\cite{Marx09-matroid} that introduced them to the area, 
as well as more recent textbooks of Cygan et al.~\cite{CyganFKLMPPS15PCbook} 
and Fomin et al.~\cite{Book_kernelization_FLSZ19}. 

We need one more operation on matroids. Let $M=(U,\cI)$ be a matroid and
$k \leq r(M)$ an integer. The \emph{$k$-truncation} of $M$ is the
matroid $M'$ where a set $I$ is independent if and only if $I \in \cI$
and $|I| \leq k$. Given a linear representation of $M$, a matrix
representing the $k$-truncation of $M$ (possibly over an extension
field of the original field) can be constructed in
polynomial time; see Marx~\cite{Marx09-matroid}
and Lokshtanov et al.~\cite{LokshtanovMPS18TALG}.

We review a few basic classes of useful matroids.
\begin{itemize}
\item Let $U$ be a ground set and $q \leq |U|$.  The \emph{uniform
    matroid} of rank $q$ over $U$ is the matroid where a set
  $I \subseteq U$ is independent if and only if $|I| \leq q$.
\item Let $G=(V,E)$ be a graph. The \emph{graphic matroid} of $G$ has
  ground set $E$, where a set $F \subseteq E$ is independent if and
  only if $F$ contains no cycles. Its dual, the \emph{co-graphic matroid},
  is the matroid where a set $F \subseteq E$ is independent if and
  only if $G$ and $G-F$ has the same number of connected components. 
\item Let $D=(V,E)$ be a digraph and $S \subseteq V$ a set of source terminals.
  A set $T \subseteq V$ is \emph{linked to $S$} if there are $|T|$
  pairwise vertex-disjoint paths from $S$ to $T$, where a path is
  allowed to have length zero (i.e., we allow $s \in S \cap T$
  and $S$ is linked to itself). Then the sets of vertices linked to $S$ 
  form the independent set of a matroid, called a \emph{gammoid}. 
\end{itemize}
All the above matroid types are representable over any sufficiently
large field, and representations for them can be efficiently computed,
although in the case of gammoids only randomized constructions are
known. Gammoids have been particularly important in previous
applications of matroids to polynomial
kernelization~\cite{KratschW20JACM}.

\subsection{Skew-symmetric matrices}

A matrix $A \in \F^{n \times n}$ over some field $\F$ is \emph{skew symmetric} if $A^T=-A$.
If $\F$ has characteristic 2, then we furthermore assume that
$A(i,i)=0$ for every $i \in [n]$. 
We will need skew-symmetric matrices for their connections to
delta-matroids (see below), but we review some facts about them
separately. 

For $X, Y \subseteq [n]$, we let $A[X,Y]$ denote the submatrix of $A$
formed of rows indexed by $X$ and columns indexed by $Y$, and
$A[X]=A[X,X]$ is the \emph{principal submatrix} indexed by $X$. 
For $i, j \in [n]$ we let $A(i,j)$ denote the value in row $i$,
column $j$ of $A$.  The \emph{support graph} of $A$ is the graph $G$
with vertex set $[n]$ and $ij \in E(G)$ if and only if $A(i,j) \neq 0$.
Note that $G$ is an undirected, simple graph.

Every skew-symmetric matrix of odd dimension is singular.  If $A$ is a
skew-symmetric matrix of even dimension, then its determinant can be
expressed through the \emph{Pfaffian}:
\[
  \det A = \paren{\Pf A}^2.
\]
The Pfaffian of $A$ can be defined as
\[
  \Pf A = \sum_M \sigma_M \prod_{ij \in M} A(i,j),
\]
where $M$ ranges over all perfect matchings of $[n]$ and
$\sigma_M \in \{1,-1\}$ is a sign term of the matching.  It is clearly
sufficient to consider terms $M$ that form perfect matchings in the
support graph of $A$. The Pfaffian of $A$ can be computed efficiently
in polynomial time.

For a skew-symmetric matrix $A$, every basis of the row or column
space of $A$ is realized by a principal submatrix, i.e., if $A$ is a
skew-symmetric matrix with rows and columns indexed by a set $V$, 
and if $B \subseteq V$ indexes a basis of the column space of $A$ then
$A[B]$ is non-singular.

An important operation on skew-symmetric matrices is \emph{pivoting}.
Let $A \in \F^{n \times n}$ be skew-symmetric and let $S \subseteq [n]$ 
be such that $A[S]$ is non-singular.  Order the rows and columns of
$A$ so that
\[
  A =
  \begin{pmatrix}
    B & C \\
    -C^T& D
  \end{pmatrix},
\]
where $A[S]=B$.  Then the \emph{pivoting} of $A$ by $S$ is
\[
  A * S =
  \begin{pmatrix}
    B^{-1} & B^{-1}C \\
    CB^{-1}& D + C^T B^{-1} C
  \end{pmatrix}.
\]
Note that this is a well-defined, skew-symmetric matrix.
It then holds, for any $X \subseteq [n]$, that
\[
  \det (A*S)[X] = \frac{\det A[X \Delta S]}{\det A[S]},
\]
where $\Delta$ denotes symmetric set difference.  In particular,
$(A*S)[X]$ is non-singular if and only if $A[X \Delta S]$ is
non-singular.

The following observation will be useful.

\begin{lemma}
  Let $A$ be a skew-symmetric matrix with rows and columns indexed
  from a set $V$ and let $B$ be a basis of $A$.  Then for all
  $u, v \in V \setminus B$, we have $(A*B)[u,v]=0$.
\end{lemma}
\begin{proof}
  Since $B$ is a basis, $A[B]$ is non-singular so the operation $A*B$ is well-defined.
  Let $(A*B)[u,v]=x$ for some  $u, v \in V \setminus B$ and assume
  towards a contradiction that $x \neq 0$. Since the diagonal is zero by
  assumption, we have $u \neq v$.  Then
  \[
    \det (A*B)[\{u,v\}] = -x^2,
  \]
  hence $\det A[B \cup \{u,v\}] = -x^2 \det A[B] \neq 0$.
  But $B \cup \{u,v\}$ is singular by assumption since $B$ is a
  basis. Hence $x=0$. 
\end{proof}

\subsection{Delta-matroids}
\label{sec:delta-matroids}
Delta-matroids are generalizations of matroids, defined by Bouchet~\cite{Bouchet87DMone}
but also independently considered under several other names; see the
survey of Moffatt~\cite{Moffatt19deltamatroids}.
A delta-matroid is a pair $(V, \cF)$ where $\cF \subseteq 2^V$ is a
non-empty collection of \emph{feasible sets},
subject to the following axiom, known as the \emph{symmetric exchange axiom}.
\begin{equation}
  \forall A, B \in \cF\, \forall x \in A \Delta B\, \exists y \in A \Delta B \colon A \Delta \{x,y\} \in \cF.\label{eq:sym-exch}
\end{equation}
Note that there is no requirement that $x \neq y$. 
Both the independent sets and the bases of a matroid satisfy the symmetric exchange axiom.
In particular, the feasible sets of a delta-matroid define the bases of a matroid if and only if
all feasible sets have the same cardinality, and the feasible sets of
maximum cardinality always form the set of bases of a matroid.

For a delta-matroid $D=(V,\cF)$ and $S \subseteq D$, the
\emph{twisting} of $D$ by $S$ is the delta-matroid
$D \Delta S = (V, \cF \Delta S)$ where
\[
  \cF \Delta S = \{F \Delta S \mid F \in \cF\}.
\]
This generalizes some common operations from matroid theory. 
The \emph{dual} delta-matroid of $D$ is $D \Delta V(D)$.
For a set $S \subseteq V(D)$, the \emph{deletion} of $S$ from $D$ 
refers to the set system
\[
  D \setminus S = (V \setminus S, \{F \in \cF \mid F \subseteq V \setminus S\}),
\]
and the \emph{contraction} of $S$ refers to
\[
  D/S = (D \Delta S) \setminus S = (V \setminus S, \{F \setminus S \mid F \in \cF, S \subseteq F\}).
\]
Both $D \setminus S$ and $D/S$ are delta-matroids if the
resulting feasible sets are non-empty.

A delta-matroid $D=(V,\cF)$ is \emph{even} if every feasible set has
the same parity (i.e., $|A| \equiv |B| \pmod 2$ for all $A, B \in \cF$;
note that there is no requirement that the sets themselves are even).
It is \emph{normal} if $\emptyset \in \cF$. 
Let $A \in \F^{n \times n}$ be a skew-symmetric matrix with rows and
columns indexed by $V$, and let
\[
  D(A)=(V, \{F \subseteq V \mid A[F] \text{ is non-singular}\}).
\]
Then $D(A)$ defines a delta-matroid. Furthermore,
$D(A)$ will be even and normal.  A delta-matroid $D$ is \emph{linear}
if $D=D(A) \Delta S$ for some $S \subseteq V(D)$.  We say that $A$
\emph{directly represents} $D$ if $D=D(A)$. Note that if $F$ is
feasible in $D(A)$, then the pivoting $A * F$ directly represents the
twisting $D(A) \Delta F$ of $D(A)$.
Representations of $D/S$ and $D \setminus S$ can be
easily produced from a representation of $D$ by combinations of
twisting/pivoting and deleting rows and columns (in the former case
assuming there is a feasible set $F$ such that $S \subseteq F$, as
otherwise $D/S$ is not a delta-matroid).
It is easy to see that a linear delta-matroid can be directly
represented if and only if it is normal. In this paper, we deal
purely with normal, linear, directly represented delta-matroids
(although an extension to the more general case is given in Section~\ref{ssec:rep-new}).

We note that every linear delta-matroid induces a (different) linear
matroid.  Specifically, let $A$ be a skew-symmetric matrix directly
representing a delta-matroid $D=D(A)$.  Let $B$ be a basis of the row
space of $A$.  Then as previously noted, $A[B]$ is non-singular, i.e.,
$B$ is feasible in $D(A)$.  Thus the matroid represented by $A$ is
precisely the matroid whose bases are the maximum-cardinality feasible
sets in $D(A)$.

As an alternative to the above representation, which may be called the
\emph{twist representation} of a linear delta-matroid, recent
work~\cite{KoanaW25stacs} has proposed a \emph{contraction representation},
which is easier to work with algorithmically. In addition, there is a
notion of \emph{projected linear delta-matroids}, which are somewhat
more general than linear delta-matroids. However, since these notions
will not be used for the majority of the paper, we defer details to
Section~\ref{ssec:rep-new}. 

\paragraph{Standard delta-matroid constructions.}
\label{sec:dm-constructions}
We review some standard constructions of linear delta-matroids.
First, let $M=(V,\cI)$ be a linear matroid represented by a matrix
$A$, and let $B$ be a basis of $M$.  Using row operations and
reordering columns, we can assume that $A$ has the form
\[
  A = (I\, N),
\]
where $B$ indexes $I$ and $V \setminus B$ indexes $N$.
Define the skew-symmetric matrix
\[
  A_B =
  \begin{pmatrix}
    O & N \\
    -N^T & O
  \end{pmatrix},
\]
where $B$ indexes the top-left part. Then $D(A_B) \Delta B$ is the
delta-matroid whose feasible sets are the bases of $M$, that is, every
linear matroid yields a linear delta-matroid.  Note that a matroid
necessarily cannot be directly represented as a delta-matroid since it
is not normal. 
The delta-matroid directly represented by $A_B$ is sometimes referred
to as a \emph{twisted matroid}.

Another important class is \emph{matching delta-matroids}.  Let
$G=(V,E)$ be a graph, and let $D(G)=(V, \cF)$ where
\[
\cF=\{S \subseteq V \mid G[S] \text{ has a perfect matching}\}.  
\]
Then $D(G)$ is a delta-matroid.  Furthermore, it is directly
representable via the Tutte matrix.  Let $V=[n]$ (w.l.o.g.)
and define a matrix $A_G$ as
\[
  A_G(i,j) =
  \begin{cases}
    x_{ij} & ij \in E(G), i<j \\
    -x_{ji} & ij \in E(G), i>j \\
    0 & ij \notin E(G),
  \end{cases}
\]
for $i, j \in [n]$, where the $x_{ij}$ are distinct formal indeterminates.
Then $D(A_G)$ is the matching delta-matroid of $G$, and it can be
represented over any sufficiently large field via the Schwartz-Zippel lemma~\cite{Schwartz80,Zippel79}.
Moreover, the Pfaffian of $A_G$, taken as a polynomial over the
variables $x_{ij}$, enumerates perfect matchings of $G$. Let $G=(V,E)$
be a graph and let $\cM$ be the set of all perfect matchings of $G$.
Let $V$ index the rows and columns of $A_G$ in the natural way.
Then
\[
\Pf A_G = \sum_{M \in \cM} \sigma_M \prod_{e \in M} x_e.
\]
Similarly, for $S \subseteq V$, $\Pf A_G[S]$ enumerates perfect
matchings of $G[S]$.

Other interesting constructions include the following. Let $D_1=(V_1,\cF_1)$ and
$D_2=(V_2,\cF_2)$ be linear delta-matroids represented over a common field, over
ground sets that are not necessarily disjoint. The \emph{delta-matroid union}
$D=D_1 \cup D_2$ is the delta-matroid $D=(V_1 \cup V_2, \cF)$
with feasible sets $\cF=\{F_1 \cup F_2 \mid F_1 \in \cF_1, F_2 \in
\cF_2, F_1 \cap F_2 = \emptyset\}$. 
The \emph{delta-sum} $D=D_1 \Delta D_2$ is the delta-matroid
$D=V_1 \cup V_2, \cF)$ with feasible sets
with feasible sets $\cF=\{F_1 \Delta F_2 \mid F_1 \in \cF_1, F_2 \in \cF_2\}$.
Both the union and delta-sum of $D_1$ and $D_2$ are linear, and a
representation can be produced in randomized linear time~\cite{KoanaW25stacs}.

For example, let $G=(V, E_R \cup E_B)$ be an edge-coloured graph where
$E_R$ contains red edges and $E_B$ contains blue edges. Let $D_1$ and $D_2$
be the matching delta-matroids of $G_R=(V, E_R)$ and $G_B=(V,E_B)$, respectively.
Then $D=D_1 \Delta D_2$ is a delta-matroid where the feasible sets
correspond to endpoints in alternating path packings in $G$. 
The delta-matroid union and delta-sum operations were introduced by Bouchet~\cite{Bouchet89dam}
and Bouchet and Schwärzler~\cite{BouchetS98dm}, although they were not
shown to retain linear representations until recently~\cite{KoanaW25stacs}.

\section{Sieving polynomial families}
\label{sec:sieving-polynomials}

We now present the framework in which the rest of the results of the
paper are produced. This is the first contribution of the paper.
While it is not ground-breaking as a reinterpretation of the linear
matroid representative sets theorem, the new perspective is essential
to the application to linear delta-matroids in Section~\ref{sec:delta-matroid-repsets}.
The method builds on basic results about non-vanishing multivariate
polynomials; we need some definitions.

Let $\F$ be a field. Let $x \in \F^q$ be a vector over $\F$ and let
$p \colon \F^q \to \F$ be a polynomial over $\F$. We use the shorthand
\[
  p(x) := p(x_1,\ldots,x_q)
\]
for the evaluation of $p$ over $x$, where $x_i$, $i \in [q]$ denote
the elements of $x$. We say that \emph{$p$ vanishes on $x$} if
$p(x)=0$, and more generally for a set of vectors $S \subseteq \F^q$,
that \emph{$p$ vanishes on $S$} if $p(x)=0$ for every $x \in S$. 
Furthermore, let $B=\{p_1,\ldots,p_r\}$ be a set of polynomials over
variables $X=\{x_1,\ldots,x_q\}$.  We say that $B$ \emph{spans} $p$ if
\[
  p(x_1,\ldots,x_q)=\sum_{i=1}^r \alpha_i p_i(x_1,\ldots,x_q)
\]
for some $\alpha_i \in \F$.  We note a standard bound on a basis for
bounded-degree multivariate polynomials. 

\begin{lemma}
  \label{lm:degree-basis}
  Let $\F$ be a field. The set of all $q$-ary multivariate polynomials
  over $\F$ of total degree at most $d$ is spanned by a set of
  cardinality at most $\binom{q+d}{d}$ (i.e., $O(q^d)$ when $d$ is a
  constant). 
\end{lemma}
\begin{proof}
  The number of monomials of total degree at most $d$ over a set of
  $q$ variables is bounded by $\binom{q+d}{d}$~\cite{LokshtanovPTWY17SODA}. 
  Clearly, the set of all such monomials over $X$ spans every
  polynomial over $X$ of degree at most $d$.  
\end{proof}

Furthermore, let $S \subseteq \F^q$ be a set of vectors and
$\rep S \subseteq S$.  We say that $\rep S$ is \emph{representative
  for $S$ with respect to $B$} if, for every polynomial $p$ spanned by
$B$ such that $p$ does not vanish on $S$, there exists a vector
$x \in \rep S$ such that $p$ does not vanish on $x$.

The following is then our basic technical tool. 

\begin{lemma}
  \label{lm:poly-nonvanish}
  Let $\F$ be a field, let $S \subseteq \F^q$ be a finite set of
  vectors for some $q$, and let $B$ be a set of $q$-ary polynomials
  over $\F$. In polynomial time in $||S||$ (i.e., the coding length of $S$),
  we can compute a set $\rep S \subseteq S$ with $|\rep S| \leq |B|$
  such that $\rep S$ is representative for $S$ with respect to $B$.
\end{lemma}
\begin{proof}
  Let $B=\{p_1,\ldots,p_r\}$ according to an arbitrary enumeration of $B$.
  For every $x \in S$, define the vector $\psi(x)=(p_i(x))_{i=1}^r$.  
  Let $Y=\{\psi(x) \mid x \in S\}$, and let
  $\rep Y \subseteq Y$ be a basis for (the space spanned by) $Y$.  Finally, let
  $\rep S \subseteq S$ contain one element $x \in S$ with $\psi(x)=y$
  for every $y \in \rep Y$.  Then clearly $|\rep S| \leq |B|$.
  Furthermore, let $p=\sum_i \alpha_i p_i$ be a polynomial spanned by $B$, 
  and let $x \in S$ be such that $p(x) \neq 0$.  
  We claim that $p$ does not vanish on $\rep S$.
  Indeed, elements $x' \in S$ such that $p(x')=0$ correspond to
  vectors $\psi(x') \in Y$ such that $\alpha \cdot \psi(x')=0$ (where $\cdot$
  denotes the dot product), and these vectors lie in a subspace of the
  space spanned by $Y$. Given that there exists a vector $\psi(x) \in Y$
  such that $\alpha \cdot \psi(x) \neq 0$, this subspace
  has a lower dimension than the whole space spanned by $Y$,
  and no basis of $Y$ can be contained in a lower-dimensional subspace.
\end{proof}

Naturally, if $B$ is not linearly independent then the bound improves
to $|\rep S| \leq \mathrm{dim}(\mathrm{span}(B))$, i.e., the size of a
basis of $B$. However, we keep the above form for simplicity. 

By Lemma~\ref{lm:degree-basis}, given $S$ and a degree
bound $d=O(1)$, we can compute a set $\rep S \subseteq S$ with
$|\rep S|=O(q^d)$ that is representative for $S$ with respect
to polynomials of degree at most $d$.

\begin{corollary}[Representative set for bounded-degree sieving polynomials]
  \label{cor:bd-nonvanish}
  Let $\F$ be a field and $S \subseteq \F^q$ a set of vectors.
  Let $d \in \N$ be given.  In polynomial time, we can compute a set
  $\rep S \subseteq S$ with $|\rep S| = O(q^d)$ such that $\rep S$ is
  representative for $S$ with respect to degree-$d$ polynomials.
\end{corollary}

In applications, assume a set $S$ of vectors is given, 
and let $\cP=\{P_1, \ldots, P_m\}$ be a collection of \emph{properties},
$P_i \subseteq S$ for each $i \in [m]$.  We say that $\cP$
admits a \emph{sieving polynomial family of degree $d$} (respectively
\emph{with basis $B$}) if for every $P_i \in \cP$ there is a polynomial
$p_i$ such that $p_i(x)\neq 0$ for $x \in S$ if and only if $x \in P_i$,
and where $p_i$ is of degree at most $d$ (respectively, $p_i$ is
spanned by $B$). Note that the collection $\cP$ does not need to be
provided for Lemma~\ref{lm:poly-nonvanish} to apply; it is sufficient
that the degree bound, respectively the basis $B$, are provided,
together with the vectors $S$. In particular, $\cP$ could
be exponentially large.

\paragraph{Weighted version.}
There is also a \emph{weighted} (or \emph{ordered}) version as
follows. Assume a weight $\omega \colon S \to \Q$ on $S$. 
Then the set $\rep S$ will contain a min-weight non-vanishing element
for every polynomial $p$ spanned by $B$. Note that this phrasing is
technically equivalent to imposing a total order $\prec$ on $S$, and
insisting that $\rep S$ contains the minimum element $x$ such that $p$
does not vanish on $x$. On the one hand, given $\omega$ we may
select any order $\prec$ on $S$ such that $\omega$ is non-decreasing in
$\prec$, and on the other hand, given $\prec$ we can fix any weights $\omega$
on $S$ increasing in the order of $\prec$. 

\begin{lemma}
  \label{lm:poly-nonvanish-order}
  Let $\F$ be a field, let $S \subseteq \F^q$, and let $B$ be a set of
  $q$-ary polynomials over $\F$.  Furthermore, assume an order $\prec$ on
  $S$.  There is a set $\rep S \subseteq S$ with $|\rep S| \leq |B|$
  such that for every polynomial $p$ spanned by $B$ either $p$
  vanishes on $S$ or there is some $x \in \rep S$ such that
  $p(x) \neq 0$ and for every $x' \in S$ with $p(x') \neq 0$
  we have $x \preceq x'$. The set $\rep S$ can be computed in polynomial
  time given $S$, $B$, and the order. 
\end{lemma}
\begin{proof}
  We reuse notation from Lemma~\ref{lm:poly-nonvanish}.
  Let $S=\{x_1,\ldots,x_n\}$ where $x_i \prec x_j$ for every $i<j$.  For
  $i \in [n]$, let $y_i=\psi(x_i)=(p_j(x_i))_{j=1}^r$ and impose the same order on
  $Y=\{y_1,\ldots,y_n\}$.  Finally, let $\rep Y \subseteq Y$ be the
  greedily computed min-basis for $Y$; i.e., $Y$ is constructed from
  $Y_0=\emptyset$, for each $i \in [n]$ letting $Y_i=Y_{i-1}+y_i$ if this set is
  independent, otherwise $Y_i=Y_{i-1}$.  Let $y \in Y \setminus \rep Y$.
  Then $Y+y$ contains a unique circuit (i.e., minimal dependent set)
  $C \subseteq Y+y$, and furthermore, $y$ is the largest element in
  $C$ according to $\prec$. As argued in Lemma~\ref{lm:poly-nonvanish},
  if $p(x) \neq 0$ for some $x \in S$ then there is also at least one
  element $x' \in \rep S$ such that $p(x')\neq 0$.  It then also
  follows that $x' \preceq x$. 
\end{proof}

We say that $\rep S$ \emph{min-represents} $S$ if $\rep S$ is the set
defined in the above lemma, with respect to some order $\prec$ on $S$. 
This is analogous to the notion defined over matroids by Fomin et al.~\cite{FominLPS16JACM}.

\subsection{Matroid representative sets via sieving polynomials}
\label{sec:l-m-appl}

We illustrate the methods by rederiving the famous results on
representative sets for linear matroids.  Let us recall the
definitions.  Let $M=(U,\cI)$ be a matroid and let $X$ be an
independent set in $M$.  We say that a set $Y \subseteq U$
\emph{extends} $X$ in $M$ if $X \cap Y = \emptyset$ 
and $X \cup Y \in \cI$, or equivalently $r(X \cup Y)=|X|+|Y|$ where
$r$ is the rank function of $M$.  

Furthermore, let $S \subseteq \binom{U}{q}$ be a set of $q$-sets of
$U$. We say that a set $\rep S \subseteq S$ is \emph{representative
  for $S$ in $M$} if, for every set $X \subseteq U$, there exists a
set $Y \in S$ extending $X$ if and only if there exists a set
$Y' \in \rep S$ extending $X$.

More generally, qualified by a bound $r \leq r(M)$, the set $\rep S$
is \emph{$r$-representative for $S$ in $M$} if the above holds only
for sets $X$ with $|X| \leq r$.  We may reduce this to the above case
by taking the $r+q$-truncation of $M$. Recall that, given a linear
representation of $M$, a linear representation of the $r+q$-truncation
of $M$ can be computed in polynomial time~\cite{Marx09-matroid,LokshtanovMPS18TALG}.

The following result is originally due to Lov\'asz~\cite{Lovasz1977},
with an algorithmic adaptation by Marx~\cite{Marx09-matroid}.  The
improved running time is due to Fomin et al.~\cite{FominLPS16JACM}.

\begin{lemma}[representative sets lemma~\cite{Lovasz1977,Marx09-matroid,FominLPS16JACM}]
  \label{lemma:repr-alg}
  Let $M=(U,\cI)$ be a linear matroid of rank $k$, and let $S \subseteq \binom{U}{q}$
  be a collection of independent sets of $M$, each of size $q$.
  Let $r \in \N$ be such that $r+q \leq k$.
  There exists a set $\rep S \subseteq S$ of size at most $\binom{r+q}{q}$
  that is $r$-representative for $S$ in $M$.  Furthermore, given a
  representation $A$ of $M$, we can find such a set $\rep S$ in
  time~$(m+||A||)^{O(1)}$, where $m=|S|$ and $||A||$ denotes the
  encoding size of $A$.
\end{lemma}

We will show a simpler version of the above result, for the
constant-degree setting, derived via the method of sieving polynomial families.
Let $M=(U,\cI)$ be a linear matroid of rank $k$, with $|U|=n$,
and let $A$ be a $k \times n$-matrix that represents $M$. 
Assume $q=O(1)$, and let a collection $S \subseteq \binom{U}{q}$ be given.
We show how to use sieving polynomial families to compute a set
$\rep{S} \subseteq S$ with $|\rep{S}|=O(k^q)$
that is representative for $S$ in $M$.

Slightly abusing notation, we treat each $Y \in S$ as a single
vector of dimension $kq$ rather than $q$ vectors of dimension $k$;
i.e., we treat the set $Y=\{u_1,\ldots,u_q\}$ 
as the vector $(u_1(1), u_1(2), \ldots, u_q(k))$. 
Introduce a set of variables $Z=\{z_{i,j} \mid i \in [k], j \in [q]\}$.
We claim that a set $\rep S \subseteq S$ which is representative for $S$ with
respect to polynomials over $Z$ of degree at most $q$ will also be
representative for $S$ in the matroid $M$. 
Concretely, for each $X \in \binom{U}{k-q}$
let $S_X=\{Y \in S \mid \text{$Y$ extends $X$}\}$.
In the terminology just established, we say that the collection of properties
\[
\cP=  \{S_X \mid X \in \binom{U}{k-q}\}
\]
admits a sieving polynomial family of degree $q$; that is,
for every $S_X \in \cP$ there is a polynomial $p$ over $Z$
of degree $q$ such that for any $Y \in S$, $p$ vanishes on $Y$
if and only if $Y \notin S_X$. 

To see this, let $X \subseteq U$ with $|X| \leq k-q$, and assume that
there exists some $Y \in S$ such that $Y$ extends $X$.  By padding
$X$, we may assume that $|X|=k-q$.  Assume
$X=\{u_1',\ldots,u_{k-q}'\}$. Now define the matrix
\[
  A_X = (X\,Z) = 
  \begin{pmatrix}
    u_1'(1) & \dots & u_{k-q}'(1) & z_{1,1} & \dots z_{1,q} \\
    \vdots &        & \vdots  &   \vdots& \vdots \\
    u_1'(k)& \dots  & u_{k-q}'(k) & z_{k,1} & \dots z_{k,q}
  \end{pmatrix}.
\]
Note that the first $k-q$ columns are constants, and the last $q$
columns contain variables,
hence $p_X(Z)=\det A_X$ defines a polynomial in $Z$. Consider some $Y \in S$,
say $Y=\{u_1,\ldots,u_q\}$, and let $p_X(Y)$ denote the evaluation of
$p_X$ where $z_{i,j} = u_j(i)$ for every $i \in [k]$, $j \in [q]$. 
Then $p_X(Y)=\det (X\, Y)$ and $Y$ extends $X$ if and only if $p_X \neq 0$.
Since clearly $\det A_X$ has degree at most $q$ in $Z$, the conclusion
follows.

It is also possible to derive the more general, fine-grained bound
$|\rep{S}| \leq \binom{k}{q}$ using a linear basis for the determinant. 
This is essentially the generalized Laplace expansion formula used
in the proof of Fomin et al.~\cite{FominLPS16JACM}; we omit the details.

\subsection{Combining sieving polynomial families}

Let us note how bounded-degree sieving polynomial families can be
combined in natural ways to form a more complex family.
The first is the simplest and most useful.

\begin{lemma}
  Let $\cP_1$ and $\cP_2$ be property families over a ground set $S$,
  and assume that they have sieving polynomials of degree $d_1$
  respectively $d_2$ and arity $r_1$ respectively $r_2$
  over a common field $\F$. Then their intersection
  \[
    \cP_1 \land \cP_2 := \{P \cap P' \mid P \in \cP_1, P' \in \cP_2\}
  \]    
  has sieving polynomials of degree $d_1+d_2$  and arity $r_1+r_2$
  over the same field.  
\end{lemma}
\begin{proof}
  Immediate. Let $P_1 \in \cP_1$ and $P_2 \in \cP_2$ be properties and
  $p_1$, $p_2$ the corresponding sieving polynomials,
  Let $\psi_1 \colon S \to \F^{r_1}$ and $\psi_2 \colon S \to \F^{r_2}$
  be vectors such that $p_i(\psi_i(x))=0$ for $x \in S$ iff $x \notin P_i$, $i=1,2$. 
  Let $\psi(x)$ be the concatenation of $\psi_1(x)$ and $\psi_2(x)$.
  Then $p(\psi(x))=p_1(\psi_1(x)) \cdot p_2(\psi_2(x))$ sieves for $P_1 \cap P_2$
  and has degree $d_1+d_2$ and arity $r_1+r_2$. 
\end{proof}

The corresponding union operation is possible without increasing the
degree. 

\begin{lemma}
  Let $\cP_1$ and $\cP_2$ be property families over a ground set $S$,
  and assume that they have sieving polynomials of degree at most $d$
  and arity $r_1$ respectively $r_2$ over a common field $\F$.
  Then their union
  \[
    \cP_1 \lor \cP_2 := \{P \cup P' \mid P \in \cP_1, P' \in \cP_2\}
  \]
  has sieving polynomials of degree $d$ and arity $r_1+r_2$,
  possibly by moving to an extension field of $\F$. 
\end{lemma}
\begin{proof}
  Let $P_1 \in \cP_1$ and $P_2 \in \cP_2$ be properties, and let
  $p_1(X_1)$ and $p_2(X_2)$ be the respective sieving polynomials.
  Let $\psi_1$ and $\psi_2$ be the vector maps.
  Let $z$ be an indeterminate value, and consider the polynomial
  $p(X_1,X_2)=p_1(X_1)+zp_2(X_2)$. Clearly if $x \notin P_1 \cup P_2$,
  then $p(\psi_1(x),\psi_2(x))=0$ for every value of $z$. On the other
  hand, for every $x \in P_1 \cup P_2$ the function $p(\psi_1(x),\psi_2(x))$
  is linear in $z$, hence there is at most one value for $z$ for
  which $p(\psi(x))=0$. By the union bound, assuming $|\F| > |S|$
  (which can be assumed by moving to an extension field) there will be
  at least one value of $z$ for every pair $(P_1,P_2)$ for which
  the resulting polynomial $p_z=p_1+zp_2$ is a valid sieving polynomial.
  Finally, since $z$ is a constant, $p_z$ has degree at most $d$.
\end{proof}

We note that here, as in the regular applications of the
representative sets lemma, we do not need to provide
the sieving polynomial families explicitly ahead of time when applying the
framework. It is enough to know the set of elements $S$, the vector map $\psi$,
the field $\F$ and the degree bound (alternatively, an explicit basis $B$).

\section{Delta-matroid representative set statements}
\label{sec:delta-matroid-repsets}

In this section, we present the main new result of the paper.  We
present two representative set statements for linear
delta-matroids, defined in terms of a set of \emph{terminals}
$T \subseteq V$. We note that these statements are somewhat more
complex and restricted in their phrasing than the matroid
representative sets statement (although they are direct generalizations
of the matroid version, up to constant factors and in the case when $q=1$).
This difference is necessary, since a direct translation of the
matroid representative sets statement to delta-matroids is not possible;
see Section~\ref{sec:rep-lb} for details.
We begin with the basic notions.

\begin{definition}[extends, repairs (delta-matroid)] \label{def:dm-extends}
  Let $D=(V,\cF)$ be a delta-matroid and $X, Y \subseteq V$.  Then we
  say that $Y$ \emph{extends} $X$ in $D$ if $X \cap Y=\emptyset$ and
  $X \cup Y \in \cF$.  If furthermore $X$ is infeasible and $X \cup Y$
  is a minimal feasible superset of $X$, then we say that $Y$
  \emph{repairs} $X$.
\end{definition}

Our basic notion of representative sets for delta-matroids is the
following.

\begin{definition}[representative set (delta-matroid)]
  Let $D=(V,\cF)$ be a delta-matroid and $T \subseteq V$ a set of
  terminals.  Let $S \subseteq \binom{V}{q}$ for some $q$.
  Then a subset $\rep S \subseteq S$ \emph{represents $S$ in $(D,T)$}
  if, for every $X \subseteq T$ such that there exists some $Y \in S$
  such that $Y$ extends $X$, there exists some $Y' \in \rep S$ such
  that $Y'$ extends $X$.
\end{definition}

With this setup, we show the following results. The first applies when
the parameter is $|T|$.

\begin{theorem}[Delta-matroid representative sets (cardinality version)]
  \label{thm:dm-extend}
  Let $D=(V,\cF)$ be a delta-matroid directly represented by a
  skew-symmetric matrix $A$.  Let $T \subseteq V$ and
  $S \subseteq \binom{V}{q}$ be given, with $k=|T|$ and $q=O(1)$.
  Then in polynomial time we can compute a set $\rep S \subseteq S$
  that represents $S$ in $(D,T)$, with $|\rep S|=O(k^q)$.
  Additionally, given an order $\prec$ on $S$, we can furthermore ensure
  that $\rep S$ min-represents $S$ in $(D,T)$.
\end{theorem}

This is a special case of a more general rank-based bound.  Let
$D=(V,\cF)$ be directly represented by $A$, and let $T \subseteq A$.
The following holds. (The polyhedral rank function is a basic rank
function of delta-matroids and is related to bisubmodular polyhedrea,
but is not important in this paper; see Bouchet and
Schwärzler~\cite{BouchetS98dm} with references.)

\begin{proposition}
  Let $T \subseteq V$.  Then the rank of the columns $T$ in $A$ equals
  \[
    r(A[\cdot, T]) = \max_{F \in \cF} |F \cap T| = \rho(T,\emptyset),
  \]
  where $\rho(P,Q)=\max_{F \in \cF} |F \cap P|-|F \cap Q|$ is the
  polyhedral rank function for $D$.
\end{proposition}

\begin{theorem}[Delta-matroid representative sets (rank version)]
  \label{thm:dm-extend-rank}
  Let $D=(V,\cF)$ be a delta-matroid directly represented by a
  skew-symmetric matrix $A$.  Let $T \subseteq V$ and
  $S \subseteq \binom{V}{q}$ be given, with $k=\rho(T,\emptyset)$ and $q=O(1)$.
  Then in polynomial time we can compute a set $\rep S \subseteq S$
  that represents $S$ in $(D,T)$, with $|\rep S|=O(k^q)$.
  Additionally, given an order $\prec$ on $S$, we can furthermore ensure
  that $\rep S$ min-represents $S$ in $(D,T)$.
\end{theorem}

In addition, we provide a full generalization of this to general
linear delta-matroids in Section~\ref{ssec:rep-new}.

\subsection{Proof (cardinality version)}

We prove Theorem~\ref{thm:dm-extend} by showing an appropriate
sieving polynomial family of bounded degree.  That is, we shall show
that there is a set of variables $Z$, and a set of evaluations
$Z=Z(Y)$ for $Y \in S$, such that for every $X \subseteq T$ there is a
sieving polynomial $p_X$ of degree at most $q$ such that $p_X$
vanishes on $Z(Y)$ if and only if $Y$ fails to extend $X$.
In this section we focus on the easier cardinality version,
leaving Theorem~\ref{thm:dm-extend-rank} for next section. 

The proof of this is very similar to the sieving
polynomials-based proof of the linear matroid representative set
theorem in Section~\ref{sec:l-m-appl}.  Let $D=(V,\cF)$ be the
delta-matroid, directly represented by a skew-symmetric matrix $A$,
and let $S \subseteq \binom{V}{q}$ and $T \subseteq V$ be given as in
Theorem~\ref{thm:dm-extend}. Define a set of variables $Z$ as
\[
 Z = \{ z_{x,i} \mid x \in T, i \in [q]\}  \cup \{ z'_{i,j} \mid 1 \leq i < j \leq q\}.
\]
Note that $|Z|=O(k)$ since $q=O(1)$ is a constant. 
Let $X \subseteq T$, and consider the matrix
\[
  A_X = \kbordermatrix{
        &   X    & [q] \\
      X & A[X]   & Z_1 \\
{}    [q] & -Z_1^T & Z_2 %
  }
\]
where $Z_1[x,i]=z_{x,i}$  for $x \in X$ and $i \in [q]$, and where
$Z_2$ is the skew-symmetric matrix with zero diagonal satisfying
$Z_2[i,j]=z'_{i,j}$ for all $1 \leq i < j \leq q$. Note that
$A_X[X]=A[X]$ consists entirely of constants while all other
off-diagonal entries of $A_X$ are simply variables.
We note that for every $Y \in S$ there is an evaluation $Z=Z(Y)$
such that $A_X(Z(Y))=A[X \uplus Y]$ (where $\uplus$ denotes disjoint
union). Given that this holds, we then simply define the polynomial
\[
  p_X(Z) = \Pf A_X
\]
and observe that $A[X \uplus Y]$ is non-singular if and only if
$p_X(Y(Z)) \neq 0$.  To finish the proof, we simply observe that
$\Pf A_X$ has degree $q$ in $Z$.

We now fill in the details. 

\begin{proof}[Proof of Theorem~\ref{thm:dm-extend}]
  Let $Z$ be as above.  Let $X \subseteq T$ be an arbitrary subset, 
  let $A_X$ be the matrix defined above, and let $p_X(Z)=\Pf A_X$.
  Let $Y \in S$, and write $Y=\{y_1,\ldots,y_q\}$ using some
  arbitrary ordering on $V$.  Consider the evaluation of $Z$
  defined by
  \begin{align*}
    z_{x,i} &= A[x,y_i]  \quad \forall x \in T, i \in [q] \\
    z'_{i,j}&= A[y_i,y_j] \quad \forall 1 \leq i < j \leq q
  \end{align*}
  Denote this evaluation by $Z(Y)$. Note that this is a well-defined
  evaluation of all variables of $Z$, and that it is independent of
  the choice of $X$. We note that with $Z=Z(Y)$,
  for all $x, x' \in X$ and $i, j \in [q]$ we have
  \begin{align*}
    A_X[x,x'] &= A[x,x'] \\
    A_X[x,i]  &= A[x,y_i] \\
    A_X[i,j]  &= A[y_i,y_j].
  \end{align*}
  Hence $A_X(Z(Y))=A[X \uplus Y]$, as claimed.
  Now, first assume that there is an item $v \in X \cap Y$, say
  $v=y_i$.  Then columns $v$ and $i$ of $A_X(Z(Y))$ are identical,
  hence $A_X(Z(Y))$ is singular. This is consistent with the fact that
  $Y$ in this case fails to extend $X$.

  Next, assume $X \cap Y = \emptyset$. Then the above may be simply
  stated as $A_X(Z(Y)) = A[X \cup Y]$ up to isomorphism
  (i.e., up to the ordering on $V$).  Then clearly $p_X(Z(Y))=0$
  if and only if $A[X \cup Y]$ is singular.  Hence $p_X$ evaluated
  at a value $Z=Z(Y)$, $Y \in S$, sieves precisely for sets $Y \in S$
  that extend $X$ in $(D,T)$.

  We are now essentially done. First, we note that the degree of $p_X$
  in $Z$ is at most $q$.  Indeed, the terms of $\Pf A_X$ ranges over
  perfect matchings $M$ in the support graph of $A_X$ over vertex set
  $X \cup [q]$, and every edge of such a matching $M$ corresponds to
  an entry $A_X[u,v]$ of $A_X$. 
  Such an entry is a constant if $u, v \in X$, and has degree 1 in $Z$
  otherwise. Since at most $q$ edges of $M$ intersect the vertex set
  $[q]$, the bound follows.  

  Hence, we have shown the existence of a sieving polynomial family of
  degree $q$. Indeed, for every $X \subseteq T$ the polynomial $p_X$
  has degree at most $q$ in $Z$, and for every $Y \in S$, $Y$ extends
  $X$ if and only if $p_X(Z(Y)) \neq 0$.  Let
  $Z_S=\{Z(Y) \mid Y \in S\}$.  By Corollary~\ref{cor:bd-nonvanish}
  we can compute a set $\rep Z_S \subseteq Z_S$ with $|\rep Z_S|=O(k^q)$ 
  that is representative for $Z_S$ up to degree-$q$ polynomials.
  We then construct a set $\rep S \subseteq S$, by arbitrarily
  including one item $Y \in S$ with $Z(Y)=Z'$ for every $Z' \in \rep Z_S$.
  The set $\rep S$ is now representative for $S$ in $(D,T)$. 
\end{proof}

If an ordering $\prec$ on $S$ is given, we may simply use
Lemma~\ref{lm:poly-nonvanish-order} in place of
Lemma~\ref{lm:poly-nonvanish} for an order-sensitive version of
Cor.~\ref{cor:bd-nonvanish}, and proceed as above.

\subsection{Proof (rank version)}

We now prove the more general version of the result, i.e.,
Theorem~\ref{thm:dm-extend-rank}. 

Let $A$ be a skew-symmetric matrix directly representing a
delta-matroid $D=(V,\cF)$. Let $T \subseteq V$ be a set of terminals
and let $k=\rho(T,\emptyset)=r(A[\cdot,T])$ be the rank of the columns
$T$ in $A$.  Let $q=O(1)$ be fixed and let $S \subseteq \binom{V}{q}$
be a set of $q$-sets over $V$.  We will show that we can compute a
representative set $\rep S$ for $S$ in $(D,T)$ with
$|\rep S|=O(k^q)$, thereby proving Theorem~\ref{thm:dm-extend-rank}. 

The proof is similar to Theorem~\ref{thm:dm-extend}. However, in order
to proceed (and not having to create a set $Z$ of $\Omega(|T| \cdot q)$
different variables) we first need a better grasp of the structure of
$A$.  

We first do away with the requirement that $X \cap Y = \emptyset$.

\begin{lemma} \label{lm:ignore-overlaps}
  Let $A$ be a skew-symmetric matrix directly defining a delta-matroid
  $D(A)=(V, \cF)$. Let $V^+=\{v^+ \mid v \in V\}$ be a set of new elements,
  and for any $U \subseteq V$ denote $U^+=\{v^+ \mid v \in U\}$.  
  We can compute a skew-symmetric matrix $A'$ over columns $V \cup V^+$,
  such that for any sets $X, Y \subseteq V$, $Y$ extends $X$
  in $D(A)$ if and only if $X \cup Y^+$ is feasible in $D(A')$. 
\end{lemma}
\begin{proof}
  Create $A'$ from $A$ by simply duplicating every element.  In other
  words, let $u, v \in V$. Then we define $A'$ so that
  \[
    A'[u^+,v^+]=A'[u^+,v]=A'[u,v^+]=A'[u,v]=A[u,v].
  \]
  Clearly, $A'$ is still skew-symmetric. Furthermore, consider sets 
  $X, Y \subseteq V$. Assume first that $Y$ extends $X$ in $A$.  Then
  $X \cap Y = \emptyset$, hence $|X \cup Y^+|=|X \cup Y|$, and
  furthermore $A[X \cup Y] = A'[X \cup Y^+]$ (up to element ordering).
  Thus $X \cup Y^+$ is feasible in $D(A')$.  On the other hand, assume
  that $X \cup Y^+$ is feasible in $D(A')$.  Assume for a contradiction
  that there is an item $v \in X \cap Y$.  Then $v$ and $v^+$ define
  two distinct, identical columns of $A'[X \cup Y^+]$, contradicting
  that $A'[X \cup Y^+]$ is non-singular. Hence $X \cap Y = \emptyset$.
  Now as above, $A'[X \cup Y^+]=A[X \cup Y]$ up to element reordering,
  and $X \cup Y$ is feasible in $D(A)$. Thus $Y$ extends $X$.
\end{proof}

Having applied this operation, we may thus replace $S$ by a set
$S^+=\{Y^+ \mid Y \in S\}$, so that for any $X \subseteq V$ and
$Y \in S$ it holds that $Y$ extends $X$ in $D(A)$ if and only if $Y^+$
extends $X$ in $D(A')$, while we are guaranteed that
$X \cap Y^+= \emptyset$. We assume this operation has been applied, but
generally drop the superscript prime symbol to ease notation,
and use the triple $(A, S, T)$ instead of $(A', S^+, T)$. 

Next, we show a pivoting operation on $A$ that uncovers more structure.

\begin{lemma}
  \label{lm:pivot-to-sparse}
  Let $B \in \cF$ be a minimal feasible set subject to having $|B \cap T|=k$.
  Then $|B| \leq 2k$. Furthermore, the pivot of $A$ by $B$ has the structure
  \[
    A*B =
    \kbordermatrix{
      & B & T \setminus B & V \setminus (T \cup B) \\
      B & \alpha & \beta & \gamma \\
      T \setminus B  & -\beta^T&   O  & O \\
      V \setminus (T \cup B) & -\gamma^T & O & \delta 
    }
  \]
  where $\alpha$ through $\delta$ are matrices of appropriate
  dimension and $\alpha$ and $\delta$ are skew-symmetric.
\end{lemma}
\begin{proof}
  Let $A'=A*B$.  To bound $|B|$, we first show $A'[B \setminus T]=0$.
  Indeed, assume $A'[u,v]=x \neq 0$ for some $u, v \in B \setminus T$.
  Then $u \neq v$ since the diagonal is zero, hence $A'[\{u,v\}]$ is
  non-singular. This implies that
  $B'=B \Delta \{u,v\} = B \setminus \{u,v\}$ is feasible.  But since
  $B \cap T = B' \cap T$, this contradicts the choice of $B$.
  Hence $A'[B \setminus T]=0$.  Next,
  let $G$ be the support graph of $A'[B]$.  Since $\emptyset$ is
  feasible for $A$, $A'[B]$ is non-singular.  In particular $G$ has a
  perfect matching. But by the above, $B \setminus T$ is an
  independent set in $G$. Hence $|B \setminus T| \leq |B|/2$, and we
  conclude $|B| \leq 2k$. 

  For the second part, we only need to show that
  $A'[T \setminus B, V \setminus B]=0$; the rest follows from
  skew-symmetry of $A'$. Assume to the contrary that
  $A'[u,v] \neq 0$ for some $u \in T \setminus B$ and
  $v \in V \setminus B$.  Then $A'[\{u,v\}]$ is
  non-singular, hence $B'=B \Delta \{u,v\} = B \cup \{u,v\}$ is a
  feasible set in $A$.  But $|B' \cap T| > |B \cap T|$, contradicting
  the choice of $B$. Hence $A'[T \setminus B, V \setminus B]=0$
  and $A'=A*X$ has the structure described in the lemma by skew
  symmetry. 
\end{proof}

We furthermore note that, having applied Lemma~\ref{lm:ignore-overlaps},
we may choose $B$ from the set $V$ of original elements of $D(A)$.
Thus we proceed assuming $B \cap Y = \emptyset$ for every $Y \in S$. 

Now, consider the matrix $A'=A*B$, and let $T'=T \cup B$ be a new
terminal set.  Consider sets $X \subseteq T$ and $Y \in S$.
By Lemma~\ref{lm:ignore-overlaps}, $Y$ extends $X$ in $D(A)$ if and
only if $X \cup Y$ is feasible in $D(A)$, and by the pivoting
operation this holds if and only if $(X \cup Y) \Delta B$ is feasible in $D(A')$.
Furthermore, since $B \cap Y=\emptyset$ we have
\[
(X \cup Y) \Delta B = (X \Delta B) \cup Y.
\]
Hence, $Y$ extends $X$ in $D(A)$ if and only if $Y$ extends
$X \Delta B$ in $D(A')$. Since $X \Delta B \subseteq T'$, it suffices
to find a set $\rep S \subseteq S$  that represents $S$ in
$(D(A'),T')$. We show that this is possible using degree-$q$ sieving
polynomials, due to the sparse structure of $A'$. 

\begin{lemma}
  Let $A'=A*B$ and $T'=T \cup B$ be as above.  In polynomial time, we
  can compute a set $\rep S \subseteq S$ that represents $S$ in
  $(D(A'), T')$ and has $|\rep S|=O(k^q)$. 
\end{lemma}
\begin{proof}
  Define a set of variables
  \[
    Z = \{z_{x,i} \mid x \in B, i \in [q]\} \cup \{z'_{i,j} \mid 1
    \leq i < j \leq q\},
  \]
  and note that $|Z|=O(k)$. Furthermore, for every $Y \in S$,
  write $Y=\{y_1,\ldots,y_q\}$ using some arbitrary ordering on $V$,
  and define an evaluation $Z(Y)$ of $Z$ defined by
  \begin{align*}
    z_{x,i} &= A'[x,y_i] \quad \forall x \in B, i \in [q] \\
    z'_{i,j}&= A'[y_i,y_j] \quad \forall 1 \leq i < j \leq q
  \end{align*}
  Next, consider some $X \subseteq T'$, and partition $X$ as
  $X=X_T \cup X_B$ where $X_B=X \cap B$ and $X_T=X \setminus B$. 
  Consider any $Y \in S$, and recall $X \cap Y = \emptyset$. 
  Then $A'[X \cup Y]$ takes the form
  \[
    A'[X \cup Y] =
    \kbordermatrix{
      & X_B & X_T & Y \\
      X_B &  A'[X_B] & A'[X_B,X_T] & A'[X_B,Y] \\
      X_T  & -A'[X_B,X_T]^T&   O  & O \\
      Y & -A'[X_BY]^T & O & A'[Y]
    }
  \]
  Define a matrix $A_X$ containing variable entries from $Z$
  such that
  \[
    A_X =
    \kbordermatrix{
      & X_B & X_T & [q] \\
      X_B &  A'[X_B] & A'[X_B,X_T] & Z_1 \\
      X_T  & -A'[X_B,X_T]^T&   O  & O \\
{}      [q] & -Z_1^T & O & Z_2
    }
  \]
  where for $x \in X_B$, $i \in [q]$ we have $Z_1[x,i]=z_{x,i}$ and
  for $1 \leq i < j \leq q$ we have $Z_2[i,j]=z'_{i,j}$, and the rest
  of $A_X$ is defined by the sparseness requirements and by
  skew-symmetry. Then it is evident that the evaluation $Z=Z(Y)$
  implies that $A_X = A'[X \cup Y]$. Hence defining $p_X(Z)=\Pf A_X$,
  for any $Y \in S$ we have $p_X(Z(Y)) \neq 0$ if and only if $Y$
  extends $X$ in $A'$. Furthermore, $p_X$ does not depend on the
  choice of $Y$ and the evaluation $Z(Y)$ does not depend on $X$.
  Finally, the argument that $p_X$ has degree at most $q$ is the same
  as in Theorem~\ref{thm:dm-extend}.
\end{proof}

This finishes the proof of Theorem~\ref{thm:dm-extend-rank}.

As with Theorem~\ref{thm:dm-extend}, if an ordering on $S$ is
additionally supplied, then we can compute a representative set
$\rep S \subseteq S$ with the same cardinality guarantee, which
furthermore contains the minimum extending set $Y \in S$
for every $X \subseteq T$.

\subsubsection{More general form}
\label{ssec:rep-new}

We now derive a more general version of Theorem~\ref{thm:dm-extend-rank},
using recent improvements in delta-matroid representations due to
Koana and the author~\cite{KoanaW25stacs}. We review the definitions.

Let $D=(V,\cF)$ be a delta-matroid. A \emph{contraction representation}
of $D$ is a skew-symmetric matrix $A$  with rows and columns indexed
by a set $V'=V \cup T$, where $T \cap V = \emptyset$,
such that $D=D(A)/T$. To distinguish the two, refer to a standard
linear representation $D=D(A') \Delta S$, $S \subseteq V$, as a
\emph{twist representation} of $D$. We can convert back and forth
between twist representations and contraction representations in
matrix multiplication time $O(n^{\omega})$, $n=|V|$. 
Contraction representations are generally easier to work with
algorithmically, as the results of~\cite{KoanaW25stacs} illustrate.

As an extension of this, let $V'=V \cup P$ and let $D'=(V',\cF')$
be a delta-matroid. The \emph{projection} $D'|P$
is the delta-matroid $D=(V,\cF)$ defined as
\[
  \cF = \{F \subseteq F \mid \exists F' \in \cF' \colon F=F' \setminus P\}.
\]
A \emph{projected linear delta-matroid} is the projection of a linear
delta-matroid. In general, a projected linear delta-matroid need not
be linear, since it may not be even. However, they retain many of the
positive properties of linear delta-matroids. In~\cite{KoanaW25stacs},
it is shown that every projected linear delta-matroid $D=(V,\cF)$ can
be converted in (randomized) polynomial time to an \emph{elementary projection}
$D=D(A')|P$ where $|P|=1$.

As an illustration, consider again a linear matroid represented
by a matrix $M$ with rows indexed by $V$. Let $M$ have rank $k$ and
assume that $M$ has full row rank. Let $T$ index the rows of $M$
and consider the matrix
\[
  A = \kbordermatrix{& T & V \\
    T & O & M \\
    V & -M^T & O}.  
\]
Then $D(A)/T$ is a contraction representation of the basis
delta-matroid represented by $M$, i.e., $F \subseteq V$ is feasible in
$D(A)/T$ if and only if $F$ is a basis of $M$, 
and $D(A)|T$ is a projection representation of the independent set
delta-matroid of $M$, i.e., $F \subseteq V$ is feasible in $D(A)|T$ if
and only if $F$ is independent in $M$. The more traditional
representation of a linear matroid as a linear delta-matroid, given
in Section~\ref{sec:prel}, can be derived from the above as
$D(A)/T=D(A') \Delta B$ for a basis $B \subseteq V$, 
where $A'=(A * (T \cup B))[V]$.

With these definitions in place, we are able to provide the general
form of the theorem.

\begin{corollary}[Delta-matroid representative sets (general version)]
  \label{cor:dm-extend-full}
  Let $D=(V,\cF)$ be a linear or projected linear delta-matroid given
  as a contraction representation $D=(D(A)/W_C)|W_P$ over a matrix $A$,
  for sets $W_C$ and $W_P$ where $|W_P| \leq 1$. Let $T \subseteq V$ and
  $S \subseteq \binom{V}{q}$ be given, with $k=\rho(T,\emptyset)$ and $q=O(1)$.
  Then in polynomial time we can compute a set $\rep S \subseteq S$
  that represents $S$ in $(D,T)$, with $|\rep S|=O(k^q)$.
  Additionally, given an order $\prec$ on $S$, we can furthermore ensure
  that $\rep S$ min-represents $S$ in $(D,T)$.
\end{corollary}

To show Cor.~\ref{cor:dm-extend-full}, we bound the size of the
contraction set $W_C$. 

\begin{claim}
  Under matrix multiplication-time operations, we can assume that
  $|W_C| \leq 2k$. 
\end{claim}
\begin{claimproof}
  Assume that $|W_C| > 2k$. 
  Let $F \subseteq V \cup W_C \cup W_P$ be a feasible set in $D(A)$
  with $W_C \subseteq F$, such that $|F \cap V|$ is minimum.  Such a
  set can be computed efficiently~\cite{KoanaW25stacs}. Furthermore,
  $|F \cap V| \leq 2k$, since the set $B$ of Lemma~\ref{lm:pivot-to-sparse} is a candidate.
  Since $\emptyset$ is feasible in $D(A)$, for every $x \in F \cap V$
  there exists $y \in F$ such that $F \setminus \{x,y\}$ is feasible
  in $D(A)$. This produces a set $F' \in \cF(D(A))$ such that
  $F' \cap V = \emptyset$, and $|F' \cap W_C| \geq |W_C|-2k$.
  We can now replace the representation $A$ by the representation
  $(A * F') \setminus (F' \cap W_C)$, where we retain the projection
  set $W_P$ (if any) and keep the contraction set $W_C'=W_C \setminus F'$
  which satisfies $|W_C'| \leq 2k$.
\end{claimproof}

Given a representation $D=(D(A) / W_C)|W_P$ where $|W_C \cup W_P|=O(k)$,
we can now simply replace the set $T$ by $T'=T \cup W_C \cup W_P$.
and apply Theorem~\ref{thm:dm-extend-rank} to $(A, S, T')$. 
Indeed, for any $X \subseteq T$ and $Y \in S$ such that $Y$ extends
$X$ in $D$, by definition there exists $Z \subseteq W_P$
such that $Y$ extends $X \cup Z \cup W_C$ in $D(A)$,
and conversely, if $Y$ extends $X \cup Z \cup W_C$ in $D(A)$
for some $X \subseteq T$ and $Z \subseteq W_P$,
then by definition $Y$ extends $X$ in $D$. 
Furthermore, the rank of $T'$ is $O(k)$. The bound follows.

\subsection{Lower bounds}
\label{sec:rep-lb}

We make some observations to note that the bounds in
Theorems~\ref{thm:dm-extend} and~\ref{thm:dm-extend-rank}
are optimal for $q=O(1)$, and that some natural extensions are
impossible.

We briefly note the distinction between our upper bounds, which take
the form $O(k^q)$, and our lower bounds below, which take the more
familiar form $\binom{k}{q}$ matching the matroid representative sets
statement. For the setting when $q=\Theta(k)$, this difference is
significant, and such a setting does occur in the algorithmic
applications of representative sets~\cite{FominLPS16JACM}.
However, we leave the question of developing tight bounds for 
Theorem~\ref{thm:dm-extend} and Theorem~\ref{thm:dm-extend-rank}
when $q=\omega(1)$ for future work. (However, for algorithmic purposes,
many of the applications of~\cite{FominLPS16JACM} can be replaced by
determinantal sieving~\cite{theoretics:14026}, if randomization is
tolerated, and determinantal sieving has recently been given a tight
generalization to the delta-matroid setting~\cite{abs-2502-13654}.)

\begin{lemma} \label{lemma:lb-repairs}
  For every $k \geq q \geq 1$, a representative set $\rep S$ for a
  directly represented delta-matroid $D(A)$, with respect to a set $T$
  of $|T|=k$ terminals and a collection $S$ of $q$-sets, may require
  $|\rep S| \geq \binom{k}{q}$, even if $A[T]=0$ and we only require
  an element $Y \in \rep S$ to be present for the case that there is a
  set $Y' \in S$ that repairs $X$, with $X \subseteq T$ and $|X|=q$.
\end{lemma}
\begin{proof}
  Let $T=\{t_1,\ldots,t_k\}$ and $U=\{u_1,\ldots,u_k\}$, and let $G$ be
  the graph on $T \cup U$ consisting of edges $t_iu_i$, $i \in [k]$.
  Let $A$ be the matching delta-matroid of $G$ and let $S = \binom{U}{q}$. 
  Note that the feasible sets in $A$ are precisely the sets
  \[
    F_I=\bigcup_{i \in I} \{t_i, u_i\}
  \]
  for subsets $I \subseteq [k]$. Furthermore $A[T]=0$ since $T$ is an
  independent set in $G$. Now, for every $Y \in S$, say $Y=\{u_i \mid i \in I\}$
  for some $I \in \binom{[k]}{q}$, let $X=\{t_i \mid i \in I\}$.
  Then $Y$ repairs $X$, but no other member of $S$ repairs (or even extends) $X$. 
  Hence every member of $S$ is required in a representative set. 
\end{proof}

We note a similar statement for extending sets when $X$ is feasible.

\begin{lemma} \label{lemma:lb-extends}
  For every $k \geq q \geq 1$ with $k$ and $q$ even, a representative
  set $\rep S$ for a directly represented delta-matroid
  with respect to a terminal set $T$ with $|T|=k$ and a collection
  $S$ of $q$-sets may require $|\rep S| \geq \binom{k}{q}$, even if
  every even subset of $T$ is feasible and we only require an element
  $Y \in \rep S$ to be present for the case that there is a set
  $Y' \in S$ that extends $X$ when $X$ is feasible and $|X|=q$,
  for $X \subseteq T$. 
\end{lemma}
\begin{proof}
  Create a graph $G$ as a complete graph $K_k$ on a vertex set $Z=\{v_1,\ldots,v_k\}$
  and with two pendants $u_i$, $w_i$ for every vertex $v_i \in Z$.
  Let $U=\{u_1,\ldots,u_k\}$, $T=\{w_1,\ldots,w_k\}$ and $S=\binom{U}{q}$.

  Let $A$ be the Tutte matrix for $G$, and consider the delta-matroid
  directly represented by the matrix $A'=(A*(Z \cup T))[U \cup T]$.
  Note that this is well defined, since $G[Z \cup T]$ has a perfect
  matching. Furthermore, a set $W \subseteq U \cup T$ is feasible in $A'$
  if and only if $|W|$ is even and there is no index $i$
  such that $W \cap \{u_i,w_i\} = \{u_i\}$. Indeed, perfect matchings
  in $G[V']$ exist for a vertex set $V'=W \Delta (T \cup Z)$
  if and only if $V'$ is even and there is no index $i$ such that
  $u_i, w_i \in V'$, which is an equivalent condition.

  Now, for every $I \in \binom{k}{q}$, consider $X=\{w_i \mid i \in I\}$.
  Then $X$ is feasible, $|X|=q$, and for any $J \in \binom{k}{q}$,
  the set $Y(J)=\{u_j \mid j \in J\}$ extends $X'$ in $D(A')$
  if and only if $I=J$.
\end{proof}

In particular, these results exclude a direct translation of the
matroid representative sets statement to delta-matroids.
That is, a set $\rep{S} \subseteq S$ that \emph{$q$-represents} $S$
in a delta-matroid $D$, in the sense that we only need to preserve
sets $Y \in \rep{S}$ extending $X$ if $|X|=q$, cannot have size
bounded as a function in $q$. Hence, the more involved statements,
of fixing a set of terminals $T$ and asking for sets $Y$ that extend
sets $X \subseteq T$, are a necessary ingredient to introduce a
meaningful ``parameter'' $k \ll n$ into the setting.

This difference fundamentally boils down to the existence of a
truncation operation for matroids, which does not exist for
delta-matroids. Indeed, this can be made precise. 
For a linear matroid $M$ represented by a matrix $A$ and a rank bound $k$
we can compute a matrix $A'$ representing $k+q$-truncation of $A$,
and working over $A'$ gives us a representative set of
size $\binom{k+q}{q}$ for $k$-representative sets over a family of
$q$-sets in a matroid. In the same way, let $D$ be a directly represented
linear delta-matroid and $q$ a constant. If there existed a directly
representable linear delta-matroid $D'$ such that for every set $S$
of cardinality at most $q$, $S$ is feasible in $D$ if
and only if $S$ is feasible in $D'$ and $D'$ is of rank at most $f(q)$
for some function of $q$, then Theorem~\ref{thm:dm-extend-rank} 
would contradict the above two lemmas and show a $q$-representative
set of some size $f(q)^{O(q)}$ not depending on $k$. 

\subsection{Examples}

Let us show a few examples to illustrate the result.

\paragraph{Linear matroid rep-sets via delta-matroids.}
First, let $M=(V,\cI)$ be a linear matroid, represented by a matrix
$A$, and let $S \subseteq \binom{V}{q}$. We will show how the
rank-based bound on a delta-matroid representative set implies the
matroid version (up to constant factors and assuming $q=O(1)$).
We follow the representation from Section~\ref{sec:dm-constructions}
with a slight simplification. 
Let $A'$ be the truncation of $A$ to rank $k+q$ and assume
$A' \in \F^{k+q \times V}$ for some field $\F$.
Introduce $k+q$ artificial elements $u_i$, $i \in [k+q]$, and let $U=\{u_i \mid i \in [k+q]\}$. 
Define the skew-symmetric matrix 
\[
  A_B =
  \kbordermatrix{
     & U & V \\
   U & O & A' \\
   V & -(A')^T & O
  }
\]
and let $D=D(A_B)$ be the delta-matroid over $U \cup V$ directly
represented by $A_B$. We note that for any $B \subseteq V$,
$B \cup U$ is feasible in $D$ if and only if $B$ is a basis of $A'$. 
Indeed, following Section~\ref{sec:dm-constructions} more closely,
consider the matrix
\[
  N=(I \, A')
\]
over columns $U \cup V$. This represents a linear matroid $M'$
of rank $k+q$ where $U$ is an artificially added basis
and $M' \setminus U$ is precisely the matroid represented by $A'$ 
(i.e., the $k+q$-truncation of $M$). By Section~\ref{sec:dm-constructions},
a set $F$ is feasible in $D$ if and only if $F \Delta U$ is a basis of $N$.
When $U \subseteq F$, this is equivalent to $U \setminus F$ being a basis of $A'$. 

The rest now follows from Theorem~\ref{thm:dm-extend-rank} applied to
matrix $A_B$, terminal set $T=U \cup V$ and the collection of $q$-sets $S$.
This produces a set $\rep{S} \subseteq S$ with $|\rep{S}| = O(k^q)$
since $q=O(1)$ and the rank of $A_B$ is at most $2(k+q)$. 
Furthermore, let $X \subseteq V$ and $Y \in S$ be such that $Y$
extends $X$ in $M$ and $|X| \leq k$. We assume $|X|=k$ by padding 
$X$ with elements from $V \setminus (X \cup Y)$ until $X \cup Y$
forms a basis of $A'$. Then for any $Y \in S$, $Y$ extends $U \cup X$
in $D$ if and only if $Y$ extends $X$ in $A'$. 

\paragraph{Matching terminals.}
Next, let us give an example based on matching delta-matroids. 
Let $G=(V,E)$ be a graph and $T \subseteq V$ a set of terminals,
$|T|=k$, such that $G-T$ has a perfect matching. Let $V'=V \setminus T$,
let $A_G$ be the Tutte matrix of $G$ and let $D$ be the delta-matroid directly
represented by $A_G * V'$; note that $A_G[V']$ is non-singular by assumption. 
Let $X \subseteq T$ be a set such that $G[V' \cup X]$ has a perfect matching.
Then for a set $Y=\{u,v\}$ with $Y \subseteq V'$ and $uv \in E$, 
$Y$ extends $X$ in $D$ if and only if there is a perfect matching of
$G[X \cup V']$ that uses the edge $uv$. Similarly, let $X=\{u,v\} \subseteq T$
be such that $G[V' \cup X]$ does not have a perfect matching.
Let $M_0$ be a perfect matching of $G[V']$.
Then a set $Y \in \binom{V'}{2}$ extends (i.e., repairs) $X$
if and only if there are $M_0$-alternating paths between $X$ and $Y$
in $G$.

\section{The Mader delta-matroid}
\label{sec:mader-delta}

Now, we present a new linear delta-matroid construction that will be
important to our main kernelization application. 

Let $G=(V,E)$ and $T \subseteq V$; let \cT be a partition of $T$.
A \emph{$T$-path} is a path with endpoints in $T$ and internal
vertices disjoint from $T$. A \emph{$\cT$-path} is a $T$-path whose
endpoints lie in different parts of $\cT$.  To avoid having to hang
a semantic difference on a detail of typography, we also refer to
a $\cT$-path as a \emph{Mader path}.  A \emph{$T$-path packing},
respectively \emph{Mader path packing}, is a collection of pairwise
vertex-disjoint $T$-paths, respectively Mader paths.
The packing is \emph{perfect} if its endpoints exhaust $T$.
In a classical result, Mader~\cite{Mader78Hpath} showed a min-max
relation characterizing the maximum cardinality of a Mader path packing.
The version below is taken from Schrijver~\cite[Chapter~73]{SchrijverBook}.
More detailed descriptions were shown by Seb\H{o} and Szeg\H{o}~\cite{SeboS04pathpacking}
and by Pap~\cite{Pap06Madermatroids}.

\begin{theorem}[Mader's $\cT$-path-packing theorem~\cite{Mader78Hpath,SchrijverBook}]
  \label{thm:mader-min-max}
  Let $G=(V,E)$ be a graph and $T \subseteq V$ a set of terminals.
  Let $\cT$ be a partition of $T$.
  The maximum size of a Mader path packing in $(G, \cT)$ equals
  \[
    \nu(G,\cT) = \min_{U_0; U_1, \ldots, U_k} |U_0| + \sum_{i=1}^k \lfloor  \frac{|B_i|}{2} \rfloor,
  \]
  where $(U_0; U_1, \ldots, U_k)$ runs over all partitions of $V$
  such that every Mader path either intersects $U_0$, or uses an edge
  spanned by some $U_i$; and where $B_i$ contains $T \cap U_i$
  as well as every vertex of $U_i$ that has a neighbor outside $U_i$
  in $G-U_0$. 
\end{theorem}

In particular, consider the graph $G'=G-U_0-\bigcup_{i=1}^k E(U_i)$.
Then $(U_0; U_1, \ldots, U_k)$ is a valid partition if and only if
every connected component of $G'$ contains terminals from at most one
part of $\cT$.

As a precursor to this result, Gallai~\cite{Gallai61} showed
the version for packing $T$-paths rather than Mader paths, i.e., the
case $\cT=\{\{t\} \mid t \in T\}$. This result is central to the
famous $O(k^2)$-vertex kernel for \textsc{Feedback Vertex Set}
by Thomassé~\cite{Thomasse10}. 
There are also multiple extensions and generalizations, such as
packing non-zero paths in group-labelled graphs~\cite{ChudnovskyGGGLS06,ChudnovskyCG08}
and the more general notion of packing \emph{non-returning} paths~\cite{Pap07,Pap08}.
Furthermore, an algorithm for packing \emph{half-integral} Mader paths 
in linear time is central to many efficient linear-time FPT algorithms~\cite{IwataYY18FOCS}.

\subsection{Mader matroids and delta-matroids}

Let $G=(V,E)$ and $\cT$ define a Mader path-packing instance, with
$T=\bigcup \cT$.  Let $\cP$ be a Mader path packing in $(G, \cT)$, and
let $S \subseteq T$ be the endpoints of paths in $\cP$. Then we say
that $\cP$ \emph{matches} $S$.  We say that $S \subseteq T$ is
\emph{Mader matchable} if there exists a Mader path packing $\cP$ that
matches $S$.  Furthermore, let
\[
\cI = \{S \subseteq T \mid \exists S' \subseteq T \colon S' \text{ is
  Mader matchable},\, S \subseteq S'\}.
\]
It can then be shown that $M=(T, \cI)$ defines a matroid, known as a
\emph{Mader matroid} (see Schrijver~\cite[Chapter~73]{SchrijverBook}).
However, a linear representation of this matroid is not immediate from
the structure description. Lovász~\cite{Lovasz80} showed that Mader
path-packings can be found via linear matroid matching, via a somewhat
intricate construction; see also Schrijver~\cite[Chapter~73]{SchrijverBook},
who asked whether every Mader matroid is a gammoid.
Subsequently, Pap~\cite{Pap06Madermatroids} studied the structure of
the min-max decomposition in Mader's theorem more deeply and settled
this in the positive.

We consider the \emph{Mader delta-matroid}, analogous to the
distinction between matching matroids and matching delta-matroids, and
show that the Mader delta-matroid is representable.
We show the following.

\begin{theorem} \label{thm:rep-delta-mader}
  Let $G=(V,E)$ be a graph, $T \subseteq V$ a set of terminals and
  $\cT$ a partition of $T$.  Then the set of Mader matchable sets
  $S \subseteq V$ in $(G, \cT)$ forms a delta-matroid $D=(T, \cF)$.
  Furthermore, a representation of $D$ can be computed in randomized
  polynomial time.
\end{theorem}

For an illustration of the difference between the Mader matroid and
the Mader delta-matroid, consider the more familiar matching matroids
and matching delta-matroids. Let $G=(V,E)$ be a graph. Then the
\emph{matching matroid} of $G$ is the matroid $M(G)=(V,\cI)$ whose bases are
the endpoints $V(M)$ of maximum matchings in $G$, and the \emph{matching
delta-matroid} is the delta-matroid $D(G)=(V,\cF)$ where a set $F$ is
feasible if and only if $G[F]$ has a perfect matching.
Then the matching matroid may sound like an intriguing concept,
but closer inspection of the Gallai-Edmonds decomposition,
which witnesses the maximum matchings of $G$, reveals that matching
matroids actually have a relatively simple structure (and indeed,
they are special cases of gammoids). On the other hand, the matching
delta-matroid, linearly represented by the Tutte matrix, appears to be
a more subtle object that cannot easily be replaced by a direct
combinatorial construction.

\begin{figure}
  \centering
  \begin{subfigure}{0.4\textwidth}
    \centering
    \begin{tikzpicture} 
      \tikzstyle{node}=[]%
  
      \node[node] (s) at (0,0) {$s$};
      \node[node,right=of s] (a) {$a$} edge (s);
      \node[node,below=of a] (b) {$b$} edge (a);
      \node[node,right=of a] (c) {$c$} edge (a);
      \node[node,right=of c] (t) {$t$} edge (c);
 
    \end{tikzpicture}
    \caption{A graph $G$ with terminals $s,t$}
    \label{a}
  \end{subfigure}\hfill
  \begin{subfigure}{0.4\textwidth}
    \begin{tikzpicture}
      \tikzstyle{node}=[]%
  
      \node[node] (s) at (0,0) {$s$};
      \node[node,right=of s] (a) {$a$};
      \node[node,below=of a] (ap) {$a'$} edge (a);
      \node[node,right=of a] (b) {$b$};
      \node[node,below=of b] (bp) {$b'$} edge[thick,red] (b);
      \node[node,right=of b] (c) {$c$};
      \node[node,below=of c] (cp) {$c'$} edge (c);
      \node[node,right=of c] (t) {$t$};
      \path[draw,thick,color=red] (s) -- (a) -- (cp) -- (t);
      \path[draw] (s) -- (ap) -- (c) -- (t);
      \path[draw] (a) -- (bp); \path[draw] (b) -- (ap);
    \end{tikzpicture}
    \caption{The transformation of $(G,\{s,t\})$.}
    \label{b}
  \end{subfigure}
  \caption{The matching transformation for $T$-packings. The
    highlighted perfect matching on the right corresponds to the $st$-path on
    the left.}
  \label{fig:gallai}
\end{figure}

As a preview of the proof of Theorem~\ref{thm:rep-delta-mader},
let us present a known reduction of the easier case of
$T$-path-packings to graph matching.
Let a graph $G$ and terminals $T \subseteq V(G)$ be given, and let $\cF \subseteq 2^T$
be the set of all terminal sets $S \subseteq T$ that are matchable by
$T$-paths in $G$. Construct an auxiliary graph $G^*$ as follows.
Write $U=V(G) \setminus T$ and let $U'=\{u' \mid u \in U\}$ be a set of
copies of vertices of $U$. We let $V(G^*)=V(G) \cup U'$, with
\begin{align*}
  E(G^*) =& \{tu, tu' \mid tu \in E(G), t\in T, u \in U\}
           \cup
 \{uv', u'v \mid uv \in E(G[U])\} \cup
           \{uu' \mid u \in U\}.
\end{align*}
This can be seen as applying the \emph{bipartite double cover}
operation on $G[U]$, after adding a loop on every vertex $u \in U$,
while leaving the terminal vertices $T$ unchanged.
An illustration is shown in Figure~\ref{fig:gallai}.
It is easy to verify that for every
$S \subseteq T$, $G^*[S \cup U \cup U']$ has a perfect matching if and
only if $S$ is matchable by $T$-paths in $G$. 

It follows from this construction that there is a delta-matroid for
$T$-path-packing, and that this is directly representable.  Indeed,
let $A^*$ be the Tutte matrix of $G^*$, and let
\[
  A=(A^**(U \cup U'))[T].
\]
In other words, we perform a contraction $D(A)=D(A^*)/(U \cup U')$. 
Then $A$ defines a delta-matroid $D(A)$ whose feasible sets are
precisely the $T$-matchable subsets of $T$. If we wish, we can also
use this for a representation of the corresponding matroid, 
by interpreting $A$ as representing a linear matroid $M(A)$ over $T$. 
This defines the matroid where a set is independent if it can be
``covered'' or ``saturated'' by a $T$-path packing, and where the
bases are precisely the $T$-matchable sets of maximum cardinality. 

We construct a representation for the Mader delta-matroid using a
similar construction, by starting from the Tutte matrix of a graph
whose edges are decorated with specific polynomials.  We show
Theorem~\ref{thm:rep-delta-mader} from the following result,
which also allows us to interrogate induced subgraphs of $G$.

\begin{theorem} \label{thm:rep-delta-mader-long}
  Let $G=(V,E)$ be a graph, $T \subseteq V$ a set of terminals and
  $\cT$ a partition of $T$.  Furthermore, let $U=V \setminus T$,
  $U'=\{v' \mid v \in U\}$, and $V^*=T \cup U \cup U'$.  Then there is
  a directly representable delta-matroid $D^*$ on $V^*$ as follows.
  For any $S \subseteq V$, the set $S^*:=S \cup \{v' \mid v \in S \cap U\}$
  is feasible in $D^*$ if and only if $S \cap T$ is Mader matchable in
  $G[S]$.  A representation for $D^*$ can be constructed in randomized
  polynomial time. 
\end{theorem}

As above, the Mader delta-matroid is then the contraction
$D=D^*/(U \cup U')$, and a representation $A$ for $D$ can be computed
from a representation $A^*$ for $D^*$ as $A=(A^* * (U \cup U'))[T]$.

\subsection{The construction}
\label{ssec:dmader-construct}
We now present the construction of the auxiliary graph. This is
essentially the same as the illustration above, but we will introduce
further algebraic cancellations in a later step. 

Let $G=(V,E)$ be a graph. To \emph{split} a set $S \subseteq V$ of
vertices in $G$ refers to the following operation, creating a new
graph $G'$.
\begin{enumerate}
\item For every $v \in S$, add a duplicate vertex $v'$ of $v$.
\item For every edge $uv \in E(G)$ with $u, v \in S$, replace $uv$
  by the two edges $u'v$, $uv'$.
  For every edge $uv \in E(G)$ with $u \notin S$ but $v \in S$,
  add the edge $uv'$.
\item Add the edge $vv'$ for every $v \in S$. 
\end{enumerate}
Let $G=(V,E)$, $T \subseteq V$ and a partition $\cT$ of $T$ define an instance of Mader
path packing,  and let $G^*$ be the graph produced by splitting
$V \setminus T$ in $G$. Write $V(G^*)=V^* \cup T$
where $V^*=V(G^*)\setminus T$.
It is well known and easy to verify that matchings in $G^*$ that saturate $V^* \cup S$
correspond to $T$-path packings in $G$ that saturate $S$ in $G$
(although not one-to-one, and not Mader path packings). 
We show something slightly sharper.  First some definitions.

\begin{definition}
  Let $G=(V,E)$, $T \subseteq V$ and a partition $\cT$ of $T$ be given.
  An \emph{oriented} $T$-path packing (respectively Mader path packing)
  is a path packing $\cP$ together with an orientation of every path
  $P \in \cP$, say with endpoints $u$ and $v$, as going either from
  $u$ to $v$ or from $v$ to $u$.  A \emph{padded} $T$-path packing
  (respectively Mader path packing) is $\cP \cup E_1 \cup E_2$,
  where $V(G) = V(\cP) \cup V(E_1) \cup V(E_2)$ partitions $V$,
  $E_1$ is a set of self-loops, and $E_2$ is a matching in $G$.
\end{definition}

An important building block of our construction is the following.
For simplicity, we assume that $G[T]$ contains no edges; clearly, this
can be assumed by subdividing edges of $G[T]$.  As another technical
simplification we also assume that all vertices of $G-T$ have
self-loops.

\begin{lemma}
  \label{lemma:char2gallai}  
  Let~$G^*$ be obtained by splitting~$V \setminus T$ in~$G$, and let~$A$ be the Tutte matrix of~$G^*$
  over a field of characteristic two. Assume that $G$ contains
  self-loops on all vertices of $V \setminus T$, and that $G[T]$ is edgeless.
  Furthermore let~$x_{uv'}=x_{u'v}$ when~$uv', u'v$ are edges in~$G^*$ not
  incident to a terminal (where $uv'$, $u'v$ are the edges in $G^*$
  corresponding to an edge $uv$ in $G$). Then
  \begin{equation}
    \label{eqn:char2gallai}
    \Pf A = \sum_{M \in \mathcal{M}} \prod_{P \in M} \prod_{v_iv_j \in P} A(i,j),
  \end{equation}
  where $M$ ranges over all padded oriented perfect~$T$-path packings in~$G$. (Here, a
  self-loop on~$v$ corresponds to the edge~$vv'$, and a matching
  edge~$uv$ in $G$ corresponds to the two edges~$uv', u'v$ in $G^*$.)
\end{lemma}
\begin{proof}
  Since $A$ is over a field of characteristic two, 
  we have $\Pf A = \sum_{M \in \mathcal{M}} \prod_{uv \in M} x_{uv}$
  where $\mathcal{M}$ ranges over all perfect matchings of $G^*$.
  Furthermore, since $x_{uv'}=x_{u'v}$ for every edge $uv \in E(G)$
  not incident with a terminal, some of these perfect matchings
  contribute identical monomials and cancel. More specifically,
  map a perfect matching $M$ in $G^*$ to an oriented subgraph in $G$
  where, for every vertex $v \in V(G) \setminus T$, the edge incident
  on $v'$ is oriented into $v$ and the edge incident on $v$ is oriented out of $v$.
  Then $M$ corresponds to an oriented subgraph $H$ of $G$ that partitions
  $V(G)$ into $T$-paths and cyclically oriented cycles on $V \setminus T$.
  Here, a loop on $v \in V\setminus T$ is considered to be a cycle.
  Furthermore, this is one-to-one, i.e., any such oriented cover of $G$
  corresponds to a distinct perfect matching of $G^*$.  
  Let $C$ be a cycle in $H$ with at least three edges, if possible.
  Then reversing the orientation of the edges in $C$ results in a distinct
  perfect matching of $G^*$ which contributes precisely the same monomial
  to $\Pf A$ as $M$. Thus we may define an involution on $\mathcal{M}$,
  by reversing the arcs of a uniquely identified cycle of $H$ (e.g.,
  the cycle $C$ with more than two edges which contains the lowest-index
  vertex).  Hence all matchings $M$ such that $H$ contains non-trivial
  cycles cancel in $\Pf A$.  On the other hand, since we have (so far)
  not made any assumption on the variables $x_{tv}$ and $x_{tv'}$
  for terminals $t \in T$, no further cancellations occur
  at this point. 
\end{proof}

We now modify the construction of Lemma~\ref{lemma:char2gallai} so that the sum runs
over (padded, oriented) Mader path packings, by introducing further
cancellations so that the set of (padded, oriented)~$T$-path packings
containing at least one non-Mader path vanishes.

We say that an orientation of a Mader path packing $\cP$ is \emph{acyclic}
if the directed multigraph it induces on the blocks of $\cT$ 
is acyclic. More precisely, let $\cP$ be an oriented $T$-path packing. 
Assume $\cT=\{T_1,\ldots,T_s\}$, and define a digraph $H$
on vertex set $V(H)=\{t_1,\ldots,t_s\}$, which contains an arc
$t_it_j$ for every path $P \in \cP$ oriented from a terminal $t \in T_i$
to a terminal $t' \in T_j$. Note that $t \neq t'$ by assumption,
and that $i \neq j$ if and only if $P$ is a Mader path.
Then $\cP$ is acyclic if and only if $H$ is acyclic.
(Note that this implies that $\cP$ is a Mader path packing, since
otherwise we get a self-loop on some vertex $t_i$, $i \in [s]$.)

\begin{figure}
  \centering
  \begin{subfigure}{0.4\textwidth}
    \centering
    \begin{tikzpicture} 
      \tikzstyle{node}=[]%
      \node[node] (s) at (0,0) {$s$};
      \node[node,right=of s] (a) {$a$};
      \node[node,below=of a] (ap) {$a'$} edge (a);
      \node[node,right=of a] (t) {$t$};
      \draw[thick,color=red] (s) -- (a);
      \draw[thick,color=red] (ap) -- (t) node [midway,below,color=black] {$z_1$};
      \draw[thick,color=blue] (s) -- (ap) node [midway,below,color=black] {$z_1$};
      \draw[thick,color=blue] (a) -- (t); 
    \end{tikzpicture}
    \caption{Terminals $s$ and $t$ are in the same partition $T_1$.
      The red and blue matchings produce identical monomials and
      cancel each other.}
    \label{Mader:a}
  \end{subfigure}\hfill
  \begin{subfigure}{0.5\textwidth}
    \begin{tikzpicture}
      \tikzstyle{node}=[]%
      \node[node] (s) at (0,0) {$t_1$};
      \node[node,right=of s] (a) {$a$};
      \node[node,below=of a] (ap) {$a'$} edge (a);
      \node[node,right=of a] (t) {$t_2$};
      \draw[thick,color=red] (s) -- (a);
      \draw[thick,color=red] (ap) -- (t) node [midway,below,color=black] {$z_2$};
      \draw (s) -- (ap) node [midway,below,color=black] {$z_1$};
      \draw (a) -- (t); 
  
      \node[node,right=of t] (s2) {$t_3$};
      \node[node,right=of s2] (a2) {$b$};
      \node[node,below=of a2] (ap2) {$b'$} edge (a2);
      \node[node,right=of a2] (t2) {$t_4$};
      \draw (s2) -- (a2); \draw (ap2) -- (t2) node [midway,below,color=black] {$z_1$};
      \draw[thick,color=red] (s2) -- (ap2) node [midway,below,color=black] {$z_2$};
      \draw[thick,color=red] (a2) -- (t2); 
    \end{tikzpicture}
    \caption{With terminal partition $T_1=\{t_!,t_4\}$,
      $T_2=\{t_2,t_3\}$, the highlighted matching produces a monomial
      that cannot be cancelled by reorientation due to its coefficient $z_2^2$.}
    \label{Mader:b}
  \end{subfigure}
  \caption{Term cancellations in the Mader delta-matroid construction.}
  \label{fig:mader}
\end{figure}

The construction and the involved term cancellations are illustrated
in Figure~\ref{fig:mader}. 

\begin{lemma}
  \label{lm:dmrepr}
  Let $A$ be the Tutte matrix defined in Lemma~\ref{lemma:char2gallai}.
  Let $\cT=T_1 \cup \ldots \cup T_k$, and let $z_1, \ldots, z_k$ be a
  set of fresh variables.  Modify $A$ such that
  for every edge $tv \in E(G)$, $t \in T$, we let $x_{tv'} = x_{tv}z_i$,
  where $t \in T_i$.  Then $\Pf A$ enumerates precisely the padded,
  acyclically oriented perfect Mader path packings. 
\end{lemma}
\begin{proof}
  Assume, as in Lemma~\ref{lemma:char2gallai}, that $G$ contains all
  self-loops and $G[T]$ contains no edges, and let $A$ be the Tutte
  matrix of $G^*$ after the modification.  By Lemma~\ref{lemma:char2gallai},
  $\Pf A$ enumerates padded oriented perfect $T$-path packings in $G$;
  this remains true after the modification.  Now consider a $T$-path
  packing which contains at least one path $P$ that has both endpoints
  in the same set $T_i$, $i \in [k]$.  Assume that $P$ is oriented to
  start with an edge $t_1u$ and end with an edge $vt_2$ in $G$,
  i.e., the matching in $G^*$ contains edges $t_1u'$ and $vt_2$
  according to the orientation convention used in Lemma~\ref{lemma:char2gallai}
  (where we may have $u=v$).
  Consider the term that only differs in reversing $P$.
  Every internal edge $pq$ in $G$, corresponding to $pq'$ respectively
  $p'q$ in $G^*$ before and after the reversal, contributes the same
  factor $x_{pq}$.  Furthermore, $t_1u'$ contributes $x_{t_1,u}z_i$,
  $vt_2$ contributes $x_{v,t_2}$, and after the reversal $t_1u$
  contributes $x_{t_1,u}$ and $v't_2$ contributes $x_{v,t_2}z_i$.
  Therefore reversing the orientation of any non-Mader path
  contributes precisely the same monomial. Hence, we can define an
  involution over the set of non-Mader $T$-path packings in $\Pf A$ by reversing 
  the orientation of some canonically chosen non-Mader path,
  leaving only Mader path packings behind.  Note that the reversal
  yields a distinct term, since $G[T]$ is edgeless.

  By a similar argument, orientations of some Mader path packing $\cP$
  where the orientations are not acyclic cancel each other, since
  reversing the orientation of all paths along a cycle on $T$ yields
  precisely the same monomial. Again, every such reversal yields a
  distinct term since $G[T]$ is edgeless.

  Finally, consider an acyclically oriented perfect Mader path
  packing $\cP$.  Then among all orientations of $\cP$, the
  contributed monomial differs only in the powers of the variables
  $z_i$ that it contributes, and by construction every path $P \in \cP$,
  oriented from $T_i$ to $T_j$, contributes a factor of $z_i$.
  We argue that every acyclic orientation of $\cP$ yields a unique
  monomial.  Indeed, since the orientation is acyclic there
  exists some $i \in [k]$ such that all paths $P \in \cP$ incident
  with $T_i$ are oriented out of $T_i$.  Hence the power of $z_i$
  in this orientation is the largest possible, and only orientations
  where similarly all $T_i$-incident paths are oriented out of $T_i$
  could possibly contribute the same monomial.  But then, we can
  repeat the argument, picking the set $T_j$ in turn, which has only
  outgoing paths except for those paths which were previously given a
  fixed orientation.  Iterating the argument yields that the monomial
  contributed by such an orientation of $\cP$ is unique.

  Since the Mader path packing is vertex-disjoint in $G$, and since
  the padding is vertex-disjoint from the Mader path packing and
  ``orientation proof'', it is easy to see that the only terms which
  could possibly contribute identical monomials correspond to distinct
  orientations of the same underlying Mader path packing. Hence the
  above covers all possible cancellations.
  In the other direction, for any acyclically oriented, perfect Mader
  path-packing $\cP$, any way we pad $\cP$ vertex-disjointly by self-loops $vv'$ and
  edge pairs $uv'$, $u'v$ on the vertices not used by $\cP$, for a term of $\Pf A$ contributing $\cP$. 
\end{proof}

We can now finish the proof of Theorem~\ref{thm:rep-delta-mader-long}.

\begin{proof}[Proof of Theorem~\ref{thm:rep-delta-mader-long}.]
  By Lemma~\ref{lm:dmrepr}, the matrix $A$, as a matrix of formal
  variables over characteristic 2, is non-singular if and only if $G$
  has a perfect Mader path packing.  Furthermore, consider a set
  $S \subset T$, and consider the matrix generated by deleting
  the rows and columns of $A$ corresponding to $T \setminus S$. 
  It can be verified that this is precisely the same matrix (up to
  variable renaming) as would be generated for Mader paths in
  $G-(T \setminus S)$.  Hence for every $S \subseteq T$, the
  matrix induced by $V(G^*) \setminus (T \setminus S)$
  is non-singular if and only if $S$ is Mader matchable in $G$. 
  A similar statement holds for sets $S^*=S \cup \{v, v' \mid v \in R\}$
  for $R \subseteq U$; deleting both rows and columns $v$ and $v'$ for
  $v \in U \setminus S^*$ yields a matrix equivalent to first deleting
  the vertex $v$ in $G$, then proceeding with the rest of the construction.
  Hence, the delta-matroid feasible sets follow from the
  self-similarity of the construction. 
  
  Note that the total degree of a monomial in $\Pf A$ is bounded by
  $O(n)$.  Hence, consider generating a matrix over GF$(2^\ell)$
  by randomly instantiating every variable in $A$.  By
  Schwartz-Zippel, using $\ell=\Omega(n)$ yields that
  the probability that any particular set $S \subseteq V$ is \emph{not} 
  correctly represented is bounded as $2^{-\Omega(n)}$.
  Hence by the union bound, $\ell=\Omega(n)$ suffices
  to yield a correct representation for every $S \subseteq V$ with
  arbitrarily high probability. We note that computation over
  GF$(2^n)$ can be done in polynomial time in $n$.
\end{proof}

Finally, to finish the proof of Theorem~\ref{thm:rep-delta-mader}, and
provide a matrix on the smaller ground set of $T$, we simply return
$A'=(A*V^*)[T]$. This is well-defined since $A[V^*]$ is non-singular,
e.g., the term $\prod_{v \in V \setminus T} x_{vv'}$ is produced only
precisely once in $\Pf A[V^*]$. Furthermore, when only
the power of Theorem~\ref{thm:rep-delta-mader} is required and
arbitrary induced subgraphs of $G$ are not of interest, it suffices
to use $\ell=\Omega(|T|+\log n)$.  This makes a minor difference to
the running time, but a potentially major difference in the total
representation size $||A'||$ of the final matrix,
since we get $||A''||=O(k^2(k+\log n))$ for $k=|T|$.
This would be suitable for a \emph{polynomial compression} (if not a
direct kernelization routine); cf.~\cite{Wahlstrom13,Book_kernelization_FLSZ19}.

\section{Mader-mimicking networks}
\label{sec:mader-mimic}

We now give the major new application of the framework.

A \emph{terminal network} be a pair $(G,T)$ where $G$ is a graph and
$T \subseteq V(G)$ a subset of the vertices referred to as \emph{terminals}. 
There has been significant interest in reducing the size of $G$
while preserving certain properties of the network, typically
approximately or precisely preserving various \emph{cut} or
\emph{flow} properties of $T$ in $(G,T)$. Variants of this concept
have been known as \emph{sparsification} or \emph{mimicking networks}.

More precisely, a \emph{mimicking network} for $(G,T)$ is a graph
$(G',T)$ such that for every bipartition $T=A \cup B$, the size of an
$(A,B)$-min cut in $G$ and $G'$ is precisely identical, while
$|V(G')|$ is bounded in terms of $|T|$. A mimicking network for
$(G,T)$ with size bounded as a function of $|T|$ always exists, but if
no further restrictions are made on the system (e.g., if the terminals
in $T$ can have arbitrarily high capacity) then the bound on $|V(G')|$
is prohibitively large, being at least exponential in
$|T|$~\cite{HagerupKNR98JCSS,KrauthgamerR13SODA,KarpovPZ19Mimic}. 
A closely related term is \emph{cut sparsifiers}, which refers to
essentially the same notion, except that we may allow the cut-sizes to
be preserved only approximately~\cite{LeightonM10spars,KrauthgamerR20planar}.

Other, more complex notions have also been considered. One interesting
case is \emph{flow sparsifiers}, which aim to (approximately) preserve
the value of any multicommodity flow system over $(G,T)$~\cite{LeightonM10spars,Chuzhoy12STOC}.
It was long open whether such a sparsifier even exists (without
compromising on the approximation ratio); this was recently answered
in the negative~\cite{KrauthgamerM23flowspars}.
Other notions closely related to our current results are
\emph{multicut-mimicking networks}~\cite{Wahlstrom22talg}
and \emph{matching mimicking networks}~\cite{EppsteinV21}.

Significantly better results are possible if we take into account the
\emph{capacity} of $T$, not only its cardinality. In particular, if
the total capacity of $T$ is $k$ (e.g., $|T|=k$ and we allow
$(A,B)$-min cuts to overlap with $A$ and $B$), then the
\emph{cut-covering lemma} implies that a mimicking network (a.k.a.\ an
exact cut sparsifier) $(G',T)$ for $(G,T)$ can be constructed with
$|V(G')|=O(k^3)$~\cite{KratschW20JACM}, and stronger results are known
on directed acyclic graphs~\cite{HeLW21dags}. Furthermore,
Chuzhoy~\cite{Chuzhoy12STOC} shows the existence of an
$O(1)$-approximate flow sparsifier of size $k^{O(\log \log k)}$.

We show that the cut-covering lemma can be extended to a mimicking
network (i.e., exact sparsifier) for \emph{Mader path-packing systems}. 
Let us introduce the definitions. 

\begin{definition}
  A \emph{Mader network} is a pair $(G, \cT)$ where $G=(V,E)$ is a
  graph and $\cT$ is a partition of a set $T = \bigcup \cT$,
  $T \subseteq V$. We say that $(G,\cT)$ is a Mader network
  \emph{over terminals $T$}. 
  A \emph{Mader-mimicking network} for $(G,\cT)$ is a
  Mader network $(G', \cT)$ such that for every $S \subseteq T$, $S$ 
  is Mader matchable in $(G,\cT)$ if and  only if $S$ is
  Mader matchable in $(G',\cT)$.
\end{definition}

Using the power of representative sets over Mader delta-matroids, we
show the following.  

\begin{theorem} \label{thm:mader-sparsification}
  Let $(G, \cT)$ be a Mader network with terminals $T=\bigcup \cT$,
  $|T|=k$. In randomized polynomial time, we can compute a
  Mader-mimicking network $(G', \cT)$ for $(G, \cT)$ with
  $|V(G')|=O(k^3)$. 
\end{theorem}

Note that by Menger's theorem, cut sparsifiers and the cut-covering
lemma correspond to a situation where $\cT$ has only two blocks.
Via closer inspection, we show the following corollary.

\begin{theorem} \label{thm:mader-sparsification-all}
  Let $(G,T)$ be a terminal network with $|T|=k$. In randomized
  polynomial time, we can compute a terminal network $(G',T)$
  such that $|V(G')|=O(k^3)$, and for every partition $\cT$ of $T$,
  the Mader path-packing numbers in $(G,\cT)$ and $(G',\cT)$
  are identical.
\end{theorem}

We find this interesting, especially in light of the stubbornly
resistant question of the existence of a polynomial kernel for
\textsc{Multiway Cut}~\cite{KratschW20JACM}.
As noted in~\cite{Wahlstrom22talg},  for edge multicuts, the
cut-covering lemma already gives a 2-approximate multicut sparsifier
with $O(k^3)$ vertices. This is not true for vertex multicuts,
since the gap between 2-way vertex cuts and vertex multiway cuts is
weaker than for edge cuts. However, there is only gap of a factor of 2
between vertex multiway cuts and Mader path-packings; hence we get the
following consequence.

\begin{theorem}
  Let $(G, T)$ be a terminal network. In polynomial time we can compute a
  2-approximate vertex multicut sparsifier $(G', T)$ for $(G,T)$
  such that $|V(G')|  = O(|T|^3)$. 
\end{theorem}

\subsection{Structural preliminaries}

Let $G=(V,E)$ be a graph, $T \subseteq V$ a set of terminals, and
$\cT$ a partition of $T$.  Recall from Theorem~\ref{thm:mader-min-max}
that the Mader path-packing number of $(G, \cT)$ is characterized by
a min-max relation over partitions $V(G)=U_0 \cup U_1 \cup \ldots \cup U_k$:
\begin{equation}
  \nu(G,\cT) = \min_{U_0; U_1, \ldots, U_k} |U_0| + \sum_{i=1}^k \lfloor  \frac{|B_i|}{2} \rfloor,
  \label{eq:mader}
\end{equation}
where $(U_0; U_1, \ldots, U_k)$ is such that every Mader path in
$(G,\cT)$ either intersects $U_0$ or uses an edge of $G[U_i]$
for some $i \in [k]$, and $B_i$ contains $T \cap U_i$ as well as
vertices of $U_i$ of non-zero degree in $G-U_0-\bigcup_i E(U_i)$.
Refer to a partition $(U_0; U_1, \ldots, U_k)$ achieving equality in (\ref{eq:mader})
as a \emph{witnessing partition}. Furthermore, for $i \in [k]$ refer
to $B_i \subseteq U_i$ as the \emph{border vertices} of $U_i$
(even though $B_i$ also contains $T \cap U_i$).
For a set $S \subseteq T$, a \emph{witnessing partition for $S$}
is a witnessing partition $(U_0; U_1, \ldots, U_k)$ of $G-(T \setminus S)$. 
Let $\nu(G,\cT_S)$ denote the Mader path-packing number of $G-(T \setminus S)$,
and let the \emph{deficiency of $S$} be $|S|-2\nu(G,\cT_S)$, i.e., the number of
unmatched terminals in a maximum Mader packing within $S$. 

We will work over a supergraph of the input graph, as follows.
Let $(G=(V,E), \cT)$ be a Mader network on terminal set $T=\bigcup \cT$,
and again write $U=V \setminus T$. For a vertex $v \in U$, to \emph{clone}
$v$ corresponds to creating a new vertex $v^+$ in $G$ such that $v$
and $v^+$ are false twins, i.e., $N(v^+)=N(v)$.  We will be working
with the supergraph $G'$ or $G$ where a clone has been added for every
non-terminal vertex of $G$.  Specifically, let $D^*$ be the
delta-matroid defined from $G'$ and $\cT$ as in
Theorem~\ref{thm:rep-delta-mader-long}, and let $A^*$ be the
skew-symmetric matrix representing $D^*$. Thus, every vertex $v \in V \setminus T$
is represented four times in $D^*$ as $v$, $v'$, $v^+$ and $(v^+)'$. 
We will work over the delta-matroid
\[
  D=D(A^* * (U \cup U'))
\]
where $U'=\{v' \mid v \in U\}$ is the set of extra elements created in
the construction of $D^*$. Thus, a set $S \subseteq T$ is feasible in $D$
if and only if it has a perfect Mader packing in $G$; for a vertex
$v \in U$, $S \cup \{v,v'\}$ is feasible if and only if $S$ has a perfect Mader
packing in $G-v$, and $S \cup \{v^+,(v^+)'\}$ is feasible if and only
if $S$ has a perfect Mader packing in $G+v^+$, i.e., after $v$ has been
cloned.

In particular, the use of clone vertices allows us to use the delta-matroid
representative sets lemmas to sieve for vertices $v$ which are
bottlenecks subject to ``almost Mader matchable'' sets $S \subseteq T$. 
We will show that such bottlenecks are well described by the above
min-max characterization, ultimately allowing us to construct a
polynomial-sized Mader-mimicking network.

We say that a graph modification, such as cloning a vertex $v$
or adding an edge $uv$, \emph{repairs} a set $S \subseteq T$ if $S$ is
not Mader matchable in $G$, but becomes Mader matchable after the
modification.  Note that in both of these cases, a set can be
repaired only if its deficiency is precisely 2. Indeed, a terminal set
of odd cardinality cannot be Mader matchable, and adding one edge or
one vertex cannot support more than one additional path. 
For brevity, we say that $S \subseteq T$ is \emph{feasible}
(respectively \emph{infeasible}) if it is Mader matchable
(respectively not Mader matchable) in $G$. 

\subsection{Mader sparsification}

Let $v \in V \setminus T$ be a non-terminal vertex. To \emph{bypass}
$v$ in $G$ refers to the following steps: first, change $N_G(v)$ into
a clique by adding any missing edges, then delete $v$.  Note that
bypassing $v$ can only increase the Mader path-packing number for any
terminal set $S \subseteq T$, as long as $v$ is not a terminal itself.
Like in previous work on vertex cut
sparsification~\cite{KratschW20JACM}, we would like to sparsify
$(G, \cT)$ by bypassing vertices $v \in V \setminus T$ one by one,
until only $O(k^3)$ vertices remain, without repairing any infeasible
sets $S \subseteq T$ in the process.

We follow this idea, but for technical reasons, we break the bypassing
operation up into smaller steps.  Recall that a \emph{simplicial
  vertex} in a graph is a vertex $v$ such that $N(v)$ is a clique.  Our
sparsification interleaves two operations: (1) delete a simplicial
non-terminal vertex, as long as there is one; (2) find a suitable
non-simplicial vertex $v$, participating in an induced $P_3$ $uvw$,
and add the edge $uw$ to $G$. We first note that the first operation
is safe. 

\begin{proposition} \label{prop:delete-simplicial}
  Let $v \in V \setminus T$ be a simplicial vertex in $G$.  Then any
  set $S \subseteq T$ is feasible in $G$ if and only if it is feasible
  in $G-v$. 
\end{proposition}
\begin{proof}
  On the one hand, clearly we do not introduce any new Mader path
  packings by deleting a vertex.  On the other hand, any Mader path
  packing for some set $S \subseteq T$ can be replaced by one not
  using $v$: Since $v \notin T$, there is a sequence $uvw$ of vertices
  passing through $v$ in any Mader path $P$ using $v$.  Then $uvw$ in
  $P$ can be replaced by $uw$ since $v$ is simplicial, and the
  resulting Mader path packing does not use $v$.
\end{proof}

Next, we show how to use cloned vertices to sieve for vertices $v$ that
may not be bypassed.

\begin{lemma} \label{claim:few-dangerous-clones}
  Let $(G, \cT)$ be a Mader network with $T=\bigcup \cT$, and let
  $S \subseteq T$ be a set that is repaired by cloning a vertex
  $v \in V \setminus T$. Then $S$ has deficiency 2 in $G$,
  and in any witnessing partition $(U_0; U_1, \ldots, U_k)$
  for $S$, either $v \in U_0$ or $v \in B_i$ for some
  set $U_i$ with $|B_i| \geq 3$. In particular, there are at most $3\nu(G,\cT)$
  vertices $v \in V \setminus T$ such that cloning $v$ repairs $S$. 
\end{lemma}
\begin{proof}
  As was observed, in order for $S$ to be repairable the deficiency
  must be positive and even, and cloning $v$ can decrease deficiency
  by at most 2.  For the rest of the statement, 
  let $U=(U_0; U_1, \ldots, U_k)$ be a witnessing partition for $S$.
  Thus
  \[
    \nu(G, \cT_S) = |S|/2-1 = |U_0| + \sum_{i=1}^k \lfloor \frac{|B_i|}{2} \rfloor,
  \]
  by the assumptions on deficiency, where $B_i$ are the border
  vertices in $U_i$.  Let $G'=G+v^+$
  where $v^+$ is a clone of $v$ (i.e., a false twin).  We note that
  the decomposition $U'=(U_0'; U_1', \ldots, U_k')$ where $v^+$ is placed
  in the same part as $v$ is a valid decomposition for $(G',\cT_S)$.
  Indeed, assume for a contradiction that there is a Mader path $P$ in
  $G'$ that does not intersect $U_0'$ and does not use an edge spanned
  by any set $U_i'$.  Then $v^+ \in V(P)$, as otherwise $P$ contradicts 
  $U$ being a valid decomposition.  If $v, v^+ \in V(P)$, then we can
  shorten $P$ to bypass $v^+$;  if $v^+ \in V(P)$ but $v \notin V(P)$
  then we can replace $v^+$ in $P$ by $v$. This creates a path $P'$
  that exists in $G$, hence either intersects $U_0 \subseteq U_0'$
  or uses an edge spanned by some set $U_i \subseteq U_i'$.

  On the other hand, since cloning $v$ repairs $S$, either $v \in U_0$
  or $v \in B_i$ for some set $U_i$, as otherwise $U'$ gives an upper
  bound of $\nu(G+v^+, \cT_S)<|S|/2$.  Assume $v \notin U_0$.  Then there
  is a set $U_i$ such that $|B_i|$ is odd and $v \in B_i$.  We claim
  $|B_i|>1$. Assume to the contrary that $B_i=\{v\}$ in $G$.  Let
  $\cP$ be a maximum Mader path packing in $G+v^+$.  By the
  decomposition theorem, every path $P \in \cP$ uses either a vertex
  of $U_0'$ or two border vertices of some set $B_i'$.  If $|B_i|=1$,
  then one of the paths in $\cP$ must use the two border vertices
  $v, v^+ \in B_i'$ (since eventually, all other pairs are exhausted).
  But as above, a path $P$ using both $v$ and $v^+$ can be shortened
  to one using only $v$, contradicting that $S$ is infeasible in $G$.
  Thus $|B_i| \geq 3$. 

  The bound on the number of ``repairing vertices'' $v$ follows.
  The worst-case ratio between the contribution of a vertex in $U_0$
  or $B_i$, $|B_i| \geq 2$ in the min-max formula is $1/3$, for the
  case that $|B_i|=3$ contributing a values of $\lfloor 3/2 \rfloor=1$. 
  Hence there are at most $3\nu(G,\cT)$ vertices $v$ as described. 
\end{proof}

As a corollary, for every infeasible set $S \subseteq T$, there are
at most $O(k)$ vertices $v$ such that cloning $v$ repairs $S$.  We
wrap up by connecting the cloning operation to our method for
bypassing $v$.  

\begin{lemma} \label{lm:bypass-equals-clone}
  Let $v \in V \setminus T$ and let $S \subseteq T$ be a set that is
  not Mader matchable.  Then $S$ is repaired by cloning $v$ if and
  only if there exists an induced $P_3$ $uvw$ in $G$ 
  such that $S$ is feasible in $G+uw$. 
\end{lemma}
\begin{proof}
  On the one hand, let $v^+$ be the clone of $v$ and assume that $S$
  is feasible in $G+v^+$.  Let $P \in \cP$ be a path in a Mader path
  packing for $S$ such that $v^+ \in V(P)$.  Let $uv^+w$ be the
  induced subpath of $P$ passing through $v^+$.  Note that this exists
  since $v \notin T$ and since $S$ is infeasible in $G$. Then
  $uw \notin E(G)$ since otherwise we could bypass $v^+$ in $P$,
  yielding a Mader path packing for $S$ in $G$. Similarly, the
  resulting path packing does exist in $G+uw$.  Conversely, if there
  is a path packing in $G+uw$ where $uvw$ is an induced $P_3$ for $v$,
  then one path in the packing must use the edge $uw$, and that path
  can be replaced in $G+v^+$ by the path containing $uv^+w$ in place
  of $uw$.
\end{proof}

We can now finish the proof. 

\begin{proof}[Proof of Theorem~\ref{thm:mader-sparsification}]
  Let $(G=(V,E), \cT=\{T_1,\ldots,T_r\})$ be given, and let
  $T=\bigcup \cT$, $|T|=k$. If there is a simplicial vertex
  $v \in V \setminus T$, then by Prop.~\ref{prop:delete-simplicial},
  we may delete $v$ in $G$ and recursively return a sparsification of
  $(G-v, \cT)$. Otherwise, we identify an appropriate induced $P_3$
  $uvw$ where $v$ is a non-terminal, add the edge $uw$ to $G$, and
  recursively restart the procedure.  We show how to do this safely.
  Refer to $v \in V \setminus T$ as \emph{dangerous} if there is an
  induced $P_3$ $uvw$ such that adding the edge $uw$ repairs some
  infeasible set $S \subseteq T$ in $G$.  We show a marking procedure
  (marking $O(k^3)$ vertices) that will mark every dangerous vertex $v$. 

  Write $U=V \setminus T$. Let $A=A^* * (U \cup U')$ and let
  $D=D(A)$ be the delta-matroid we have constructed, over ground set
  $V \cup U' \cup U^+ \cup U^{++}$, where $U'=\{v' \mid v \in U\}$,
  $U^+=\{v^+ \mid v \in U\}$ and $U^{++}=\{(v^+)' \mid v \in U\}$.
  For $v \in U$, write $Y(v)=\{v^+,(v^+)'\}$ and
  let $S=\{Y(v) \mid v \in U\}$.
  By Theorem~\ref{thm:dm-extend}, there exists an efficiently
  computable sieving polynomial family for $(D, T, S)$, 
  of arity $r=O(k)$ and degree $2$. Initialize $Z=\emptyset$ and
  repeat the following $3k$ times:\footnote{For an equivalent result,
    we could work with a single sieving polynomial family according to
    a less ad-hoc recipe as follows.  Let $\F$ be the field over which
    $A$ is constructed, and assume $|\F| > |U|$, or otherwise
    replace $\F$ by an extension field with more than $|U|$ elements.
    Let $M$ be the uniform matroid over $U$ of rank $3k$
    and construct a linear representation of $M$ over $\F$.  This is
    possible in polynomial time since $|\F| > |U|$~\cite{OxleyBook2}.
    Let $S \subseteq T$ be an infeasible set such that $Y(v)$ repairs $S$,
    and let $Q \subseteq U$ consist of all vertices $u \neq v$ such
    that $Y(u)$ repairs $S$.  Then there is a quadratic polynomial
    $p_1$ that sieves for $v \in (Q+v)$, and a linear polynomial $p_2$
    that sieves for $v \notin Q$; hence their product is a cubic
    polynomial that sieves only for the vertex $v$. Thus the loop can
    be replaced by a single call for a representative for a cubic
    sieving polynomial family.
  }
  \begin{enumerate}
  \item Compute a representative set $\rep S \subseteq S$ for $S$
    in $(D,T)$.
  \item Add $\{v \in U \mid Y(v) \in \rep S\}$ to $Z$
  \item Update $S \gets S \setminus \rep S$
  \end{enumerate}
  Upon completion, if $Z=U$, declare the procedure to have terminated
  successfully; otherwise select an arbitrary vertex
  $v \in U \setminus Z$, find an induced $P_3$ $uvw$ in $G$ centred on
  $v$, and return a recursively computed sparsification for $(G+uw, \cT)$.

  The construction will complete in polynomial time.  Indeed, every
  step of the above loop can be computed in polynomial time, and the
  recursion depth is bounded by $n+(n^2-|E|)$ since every step either
  reduces $n=|V(G)|$ or removes a non-edge from $G$.  The process
  terminates in a Mader network $(G', \cT)$, since no vertex in $T$ is
  removed, and $|V(G')| \leq |T|+|Z|=O(k^3)$ by Theorem~\ref{thm:dm-extend}.
  We prove that $(G', \cT)$ is a Mader-mimicking network of $(G,\cT)$, i.e., 
  for every $S \subseteq T$, $S$ is feasible in $G$ if and only if $S$
  is feasible in $G'$. We prove the statement by induction
  on $n$ and the number of non-edges of $G$.

  As base case, consider the case that $V(G)=Z \cup T$ and $G$ has no
  simplicial non-terminal vertices.  Then we return $(G,\cT)$
  unchanged, and the correctness is trivial.  If $G$ contains a
  simplicial non-terminal vertex $v$, then $(G-v, \cT)$ is a Mader-mimicking
  network for $(G, \cT)$ by Prop.~\ref{prop:delete-simplicial} and
  the return value is correct by the inductive hypothesis.  Otherwise,
  our algorithm makes a recursive call on some input $(G+uw, \cT)$.
  Now by the inductive hypothesis, it is sufficient to prove that
  $(G+uw, \cT)$ is a Mader-mimicking network for $(G, \cT)$.  Assume not, and
  let $S \subseteq T$ be a set which is infeasible in $G$ but feasible
  in $G+uw$.  Since we added the edge $uw$, the algorithm must have
  selected an induced $P_3$ $uvw$ where $v \in V \setminus (T \cup Z)$.
  By Lemma~\ref{lm:bypass-equals-clone}, cloning $v$ in $G$ repairs $S$.
  By construction, this is equivalent to $S \cup Y(v)$ being feasible
  in $D$, i.e., $Y(v)$ repairs $S$ in $D$.  If $v \notin Z$, then $Z$ contains at
  least $3k$ distinct vertices $x \neq v$ such that cloning $x$ in $G$
  repairs $S$. But by Lemma~\ref{claim:few-dangerous-clones},
  there are at most $3k$ such vertices in total.  Hence we would have
  $v \in Z$, in contradiction to our assumptions. 
  Hence $v$ is not dangerous, and $(G+uw, \cT)$ is a Mader-mimicking
  network for $(G,\cT)$.  By the inductive hypothesis, the output
  $(G', \cT)$ is a Mader-mimicking network for $(G,\cT)$. 
\end{proof}

As a final technical note, we note that there are situations where a
terminal set $S$ of deficiency 2 is feasible after cloning a vertex
$v$ \emph{twice}, but not after cloning $v$ only once.  

\begin{example}
  Create a $K_4$ on vertex set $\{a,b,c,d\}$. Attach three terminals
  $t_{a,i}$ to $a$, one terminal $t_b$ to $b$, one terminal $t_c$
  to $c$ and one terminal $t_d$ to $d$. Let the partition be
  $\cT=\{\{t_{a,1}, t_{a,2},t_{a,3}\}, \{t_b\}, \{t_c\}, \{t_d\}\}$.
  The Mader path-packing number is 2, hence the deficiency of the terminal set is 2.  The graph
  $G+a^+$ still has only packing number 2 (as two paths exhaust four 
  out of the five vertices $a, a^+, b, c, d$), but $G+a^++a^{++}$ has
  packing number 3.
\end{example}

Following this example, the appealing alternative construction of
simply finding an irrelevant vertex $a$ and bypassing it may not work, 
since there could be situations in which a vertex $a$ should not be
bypassed yet we cannot guarantee that $a \in Z$ after the above
computation.  Hence the need for the above, more careful modification
of adding edges of $N(a)$ one at a time.

\subsection{Applications}

We show two implications of this. First, we modify the construction to
give a mimicking network for Mader path packings with respect to all
partitions of the terminal set, as opposed to only considering the
original partition $\cT$.

\begin{theorem} \label{thm:mader-all-partitions}
  Let $G=(V,E)$ be an undirected graph and $T \subseteq V$ a set of
  terminals, $|T|=k$. In polynomial time we can compute a graph $G'=(V',E')$, 
  $T \subseteq V'$, with $|V|=O(k^3)$ such that for every partition $\cT$ of $T$, the
  Mader path packing numbers of $(G,\cT)$ and $(G',\cT)$ are identical.
\end{theorem}
\begin{proof}
  Given $G$ and $T$, for every terminal $t \in T$ and every $i \in [k]$
  let $t^i$ be a pendant of $t$, and let $\cT=\{\{t^i \mid t \in T\} \mid i \in [k]\}$
  be the partition where ever block consists of a set of pendant
  copies of the original terminals $T$. Let $T'=\bigcup \cT$. 
  Let $G'$ be the resulting graph. Since $|T'|=k^2$, a bound of
  $|V'|=O(k^6)$ is immediate by invoking Theorem~\ref{thm:mader-sparsification} on $(G', \cT')$. 
  We show the tighter bound by utilising the rank version (Theorem~\ref{thm:dm-extend-rank})
  of the representative sets statement and analysing the process in
  Theorem~\ref{thm:mader-sparsification} more carefully.
  Let $D=D(A)$ be the delta-matroid constructed in the proof of
  Theorem~\ref{thm:mader-sparsification}, and let us 
  recall the steps of the representation of $D(A)$.
  Let $U=V \setminus T$ and let $G^+=(V^+,E^+)$ be the graph resulting
  from cloning every vertex $v \in U$, creating a vertex set
  $V^+=V \cup U^+$, $U^+=\{v^+ \mid v \in U\}$,
  and let $G^*$ be obtained by splitting $U \cup U^+$. 
  Thus $G^*$ now contains four copies $U$, $U^+$, $U'$, $U^{++}$
  of the vertex set $U$, where $U^{++}$ contains the split copies
  $(v^+)'$ of the clones $v^+$, $v \in U$.
  Then $D^*=D(A^*)$ is defined via a particular (non-generic)
  evaluation of the Tutte matrix of $G^*$ as in Section~\ref{ssec:dmader-construct},
  and $D=D(A)$ where $A=(A^* * (U \cup U'))$ is a twist of $D^*$. 
  
  \begin{claim} \label{claim:rank}
    The set $T'$ has rank at most $4k$ in $D$. 
  \end{claim}
  \begin{claimproof}
    We have $N_{G^*}(T')=\{t, t', t^+, (t^+)' \mid t \in T\}$,
    and $G^*[T']$ is edgeless. Hence any matching (whether perfect or
    not) in $G^*$ intersects at most $4k$ vertices of $T'$. 
    Since $A^*$ is a specific evaluation of the Tutte matrix of $G^*$,
    and $A$ is a twist of $A^*$ that does not affect $T'$,
    the same rank bound holds in $D$. 
  \end{claimproof}

  We can now proceed as Theorem~\ref{thm:mader-sparsification}, and define
  $Y(v)=\{v^+,(v^+)'\}$ for $v \in U$ and let $S=\{Y(v) \mid v \in U\}$.
  Let $\rep{S} \subseteq S$ be the representative set for $S$ in $(D,T')$
  and use Theorem~\ref{thm:dm-extend-rank} to compute $\rep{S}$. 
  By Claim~\ref{claim:rank}, $|\rep{S}|=O(k^2)$. 
  Next, we step through the process of marking vertices in $U$
  to place into a set $Z$, except instead of using $3k$ iterations
  as in Theorem~\ref{thm:mader-sparsification}, we use $6k$ iterations.
  We claim that this is sufficient. Indeed, the bound $3k$ in
  Theorem~\ref{thm:mader-sparsification} comes from
  Lemma~\ref{claim:few-dangerous-clones} using Mader
  path-packing number $k$. Since the Mader path-packing number of $(G', \cT')$
  is $2k$ (consisting of paths such as $t^1tt^2$ for every $t \in T$),
  by Lemma~\ref{claim:few-dangerous-clones} there are at most $6k$ ``dangerous'' vertices in any set
  $R \subseteq T'$ such that $R$ can be repaired by adding an edge to $G'$.
  Proceed as in Theorem~\ref{thm:mader-sparsification}, either adding
  an edge $uw$ corresponding to an induced $P_3$ $uvw$ or deleting a
  simplicial vertex, while making sure that no edges are added
  incident to $T'$ by only acting on non-simplicial vertices outside
  of $T$. This ensures that at every step of the way,
  the argument that $T'$ has rank $4k$ remains valid (in particular,
  all graph modifications performed correspond to modifications on the
  original graph $G$). Let $H$ be the graph produced when this is no
  longer possible. Then $(H, \cT')$ is a Mader-mimicking network of $(G',\cT')$
  with $|V(H)| = O(|T'|+|T|+|Z|)=O(k^3)$. Now let $H'=H-T'$. We claim that
  $(H,T)$ is a valid sparsification output of $(G,T)$. 
  
  \begin{claim}
    Let $\cT$ be a partition of $T$. Then $\nu(G,\cT)=\nu(H,\cT)$. 
  \end{claim}
  \begin{claimproof}
    Write $\cT=\{T_1, \ldots, T_p\}$. 
    Let $S \subseteq T$ be the endpoints of a maximum Mader path
    packing in $(G,\cT)$. Define
    \[
      S' = \{t^i \mid t \in S \cap T_i, T_i \in \cT\} \cup
      \cup \{t^1, t^2 \mid t \in (T \setminus S)\}.
    \]
    Let $\cP$ be a Mader path packing in $G$ with endpoints $S$. 
    We can extend $\cP$ to a Mader path packing in $G'$
    by adding a path $t^1tt^2$ for every $t \in T \setminus S$
    and extending every path $P \in \cP$, with endpoints
    $t_a \in T_i$, $t_b \in T_j$, with edges $t_a^it_a$ and $t_bt_b^j$.
    Clearly this defines a perfect Mader path packing on endpoints
    $S'$ in $G'$. Since $(H', \cT')$ is a Mader-mimicking network
    for $(G', \cT')$, there also exists a perfect Mader path packing
    $\cP'$ for $S'$ in $H'$. For any $t \in T \setminus S$,
    $\cP'$ must again contain the path $t^1tt^2$, since
    there are no other edges incident on $t^1$ and $t^2$
    and two distinct paths cannot use the same vertex $t$.
    Thus $\cP'$ also contains a set of Mader paths $P$
    that match the endpoints $S$, with respect to the partition $\cT$.
    Thus $S$ is Mader matchable in $H$.

    In the other direction, let $\cP$ be a maximum Mader path packing
    in $(H,\cT)$ and let $S \subseteq T$ be the endpoints of $\cP$. 
    Again define
    \[
      S' = \{t^i \mid t \in S \cap T_i, T_i \in \cT\} \cup
      \cup \{t^1, t^2 \mid t \in (T \setminus S)\}
    \]
    and extend $\cP$ to a Mader path packing on $S'$
    by adding paths $t^1tt^2$ for $t \in T \setminus S$. 
    Since $(H',\cT')$ is a Mader-mimicking network for $(G',\cT')$,
    the set $S'$ is also Mader matchable in $(G', \cT')$.
    Let $\cP'$ be a Mader path packing in $(G',\cT')$ with endpoints $S'$.
    Again for any $t \in T \setminus S$, the only possible path in $\cP'$
    is $t^1tt^2$, hence no other path in $\cP'$ intersects vertices of
    $T \setminus S$. The remaining paths in $\cP'$ must connect
    some terminals $t_a^i$ and $t_b^j$, where $t_a \in T_i$ and $t_b \in T_j$,
    and each such path must also connect $t_a$ with $t_b$. 
    Hence $S$ is Mader matchable in $G$ with respect to partition $\cT$.     
  \end{claimproof}
  
  Since all steps can be performed in polynomial time, this finishes
  the proof. 
\end{proof}

We show a further application of this to sparsification of the
\textsc{Vertex Multicut} problem. Let $G=(V,E)$ be an undirected graph
and $T \subseteq V$ a set of terminals in $G$. As in~\cite{Wahlstrom22talg},
for a partition $\cT$ of $T$, let a \emph{vertex multiway cut for $\cT$}
be a set of vertices $X \subseteq V$ such that every connected
component of $G-X$ contains terminals from at most one part of $\cT$. 
Similarly, given a set of cut requests $R \subseteq \binom{T}{2}$,
a \emph{vertex multicut for $R$} is a set of vertices $X \subseteq V$
such that every connected component of $G-X$ contains at most one
member of every pair $st \in R$. 
A \emph{vertex multicut-mimicking network} for $(G,T)$ is a graph
$G'=(V',E')$ with $T \subseteq V'$, such that for every set of cut
requests $R \subseteq \binom{T}{2}$, the size of a minimum multicut
for $R$ is equal in $G$ and in $G'$. Finally, $(G',T)$ is a
\emph{$q$-approximate multicut sparsifier} for $(G,T)$ if for every set
of cut requests $R \subseteq \binom{T}{2}$, the size of a minimum
multicut for $R$ differs by at most a factor of $q$ in $G$ and in $G'$.
For the \emph{edge cut}
versions of these notions, in order to get a 2-approximate multicut sparsifier
it suffices to preserve all \emph{minimum edge cuts} over $T$,
i.e., to preserve the size of an $(A,B)$-min cut for every bipartition
$T=A \cup B$ of $T$~\cite{Wahlstrom22talg}. This essentially follows since,  for every
partition $\cT=\{T_1, \ldots, T_p\}$, there is only a factor of 2
between the minimum edge multiway cut for $\cT$ and the sum of
\emph{isolating min-cuts} $(T_i, T \setminus T_i)$, $I \in [p]$.
Furthermore, multicuts for cut requests $R \subseteq \binom{T}{2}$
naturally reduces to multiway cuts for \emph{some} partition $\cT$
of $T$, according to the connected components in $G-X$ where $X$ is
the optimal solution. Such a 2-approximate multicut sparsifier
with $O(|T|^3)$ vertices can therefore be computed in randomised
polynomial time via the cut-covering lemma~\cite{KratschW20JACM}.
It is also known
that an exact multicut-mimicking network with $|T|^{\log^{O(1)} |T|}$
vertices can be computed in randomised polynomial time,
via more intricate methods~\cite{Wahlstrom22talg}.

We note that for \emph{vertex multiway cuts}, it no longer suffices to
consider only isolating min-cuts (or more generally, vertex min-cuts
for bipartitions $T=A \cup B$ of $T$). However, there is a natural
duality between vertex multiway cuts and Mader path-packings.
We include a proof for completeness. 

\begin{lemma} \label{lemma:mader-gives-apx}
  Let $G=(V,E)$ be an undirected graph and $\cT=T_1 \cup \ldots \cup T_s$
  a partition of a set $T=\bigcup_i T_i$. Let $\nu=\nu(G,\cT)$ be the
  Mader path-packing number of $(G,\cT)$. There is a vertex multiway cut 
  $X$ for $(G,\cT)$ of size at most $|X| \leq 2\nu$ (where $X$ may
  intersect $T$).
\end{lemma}
\begin{proof}
  Recall the min-max formula of  Mader's $\cT$-path-packing theorem (Theorem~\ref{thm:mader-min-max}),
  and let $(U_0; U_1, \ldots, U_p)$ be a partition of $V$ that witnesses the value of $\nu(G,\cT)$.
  That is, $V = U_0 \cup \ldots \cup U_p$, and every $\cT$-path either intersects $U_0$
  or uses an edge of $G[U_i]$ for some $i \in [p]$. Furthermore, for $i \in [p]$
  let $B_i \subseteq U_i$ contain $T \cap U_i$ together with the
  vertices $v \in U_i$ that neighbour some vertex of $V \setminus (U_0 \cup U_i)$. 
  Then by Theorem~\ref{thm:mader-min-max},
  \[
    \nu(G,\cT) =  |U_0| + \sum_{i=1}^p \lfloor \frac{|B_i|}{2} \rfloor.
  \]
  Construct a set $X$ by starting from $X=U_0$, then for every $i \in [p]$
  select an arbitrary member $x \in B_i$ and add $B_i-x$ to $X$. We get
  \[
    |X| = |U_0| + \sum_{i=1}^p (|B_i|-1) \leq 2\nu(G,\cT),
  \]
  since $\lfloor \frac{n}{2} \rfloor \geq (n-1)/2$ for every $n \in \N$.
  It remains to show that $X$ is a multiway cut for $\cT$.
  Consider a $\cT$-path $P$ in $G$. If $P$ intersects $U_0$ then clearly $P$ also intersects $X$.
  Otherwise, there is an index $i \in [p]$ and an edge $e$ in $G[U_i]$
  such that $P$ uses $e$. Let $P'$ be the maximal subpath of $P$ that is contained in $U_i$ and contains $e$,
  and let $u$ and $v$ be the endpoints of $P'$. Then $u \neq v$ since $P'$ is non-empty.
  We claim that $u, v \in B_i$. Indeed, consider $u$. Either $u$ is an endpoint of $P$,
  in which case $u \in T$, or $P$ contains an edge $wu$ where $w \in U_j$ for some $j \in [p]$,
  $j \neq i$ (since $P$ does not intersect $U_0$). Hence $u \in B_i$, and symmetrically $v \in B_i$.
  Thus $X$ contains at least one of $u$ and $v$, and $X$ hits $P$.
\end{proof}

\begin{theorem}[Cor.~\ref{cor:two-approx-vmc}]
  Let $(G, T)$ be a terminal network. In polynomial time we can compute a
  2-approximate vertex multicut sparsifier $(G', T)$ for $(G,T)$
  such that $|V(G')|  = O(|T|^3)$. 
\end{theorem}
\begin{proof}
  Let $(G', T)$ be the ``all partitions Mader-mimicking network'' of $(G,T)$
  computed by Theorem~\ref{thm:mader-all-partitions}. Let $k=|T|$.
  We show that $(G',T)$ is a 2-approximate vertex multicut sparsifier
  for $(G,T)$. Let $R \subseteq \binom{T}{2}$ be a set of cut requests. 
  On the one hand, the size of a multicut for $R$ in $G'$ is at least
  as high as in $G$, since the only reductions employed in creating $G'$
  are to either delete a non-terminal simplicial vertex or add an
  edge, and neither of these can reduce the size of a minimum multicut.
  In the other direction, let $X$ be a minimum multicut for $R$ in $G$
  and let $\cT$ be the partition of $T$ induced by connected
  components of $G-X$, placing every terminal $t \in X \cap T$
  into a separate partition $\{t\}$. Then $\nu(G,\cT)=\nu(G',\cT) \leq |X|$
  (since every Mader path of $\cT$ connects a pair of terminals
  separated by $X$ in $G$, and the paths are pairwise vertex-disjoint).
  By Lemma~\ref{lemma:mader-gives-apx}, there therefore exists a
  vertex multiway cut $X'$ for $\cT$ in $G'$ with $|X'| \leq 2\nu(G', \cT)$.
  Clearly $X'$ is also a multicut for $R$. This finishes the proof. 
\end{proof}

\section{Conclusions and open questions}
\label{sec:conc}

We presented several new results, comprising a framework towards new
results in polynomial kernelization. First, we introduced the notion
of \emph{sieving polynomial families}, generalizing the representative
sets lemma for linear matroids as well as the method of instance
sparsification based on bounded-degree polynomials~\cite{JansenP19toct}.
In particular, this provides a simpler proof of the representative
sets lemma for linear matroids, bypassing the need for exterior algebra
or generalized Laplace expansion~\cite{Lovasz1977,Marx09-matroid,FominLPS16JACM}
(at least for the case that the cardinality $q$ of the set family that
is being reduced is a constant). 

Second, using this approach we establish representative sets-like
statements for linear \emph{delta-matroids}, which is a setting that
generalises linear matroids while also capturing other structures.
These representative set statements are somewhat more intricate than
the representative-sets lemma for linear matroids, since
delta-matroids appear to be somewhat more fragile structures compared
with matroids (in particular, there is no \emph{truncation} operation).
However, again up to constant factors and the assumption that
$q=O(1)$, our results form a strict generalisation of the linear
matroid representative sets lemma.

Third, we show a new class of linear delta-matroids, \emph{Mader delta-matroids},
which captures the problem of packing vertex-disjoint $\cT$-paths
over a partition $\cT$ of a terminal set $T$; i.e., a delta-matroid
capturing the setting of Mader's remarkable $\cS$-paths
theorem~\cite{Mader78Hpath,SchrijverBook}. 
This allows us to generalize the $O(k^3)$-vertex exact cut sparsifier
results of Kratsch and Wahlström~\cite{KratschW20JACM}
(i.e., the cut-covering sets lemma) 
to a notion of exact \emph{Mader-mimicking networks} with $O(k^3)$ vertices.
In particular, this implies a 2-approximate vertex multicut sparsifier
of $O(k^3)$ vertices, generalising a similar result for edge
multicuts~\cite{Wahlstrom22talg}. 

\paragraph{Questions: Delta-matroid representative sets.}
In the extended abstract of this article~\cite{Wahlstrom24repset},
several gaps were noted in the derived representative set statements
for delta-matroids. Some, but not all, of these have been filled using
recent methods~\cite{KoanaW25stacs}; we note this for reference. 

First, in the extended abstract~\cite{Wahlstrom24repset}
results were only given for directly represented delta-matroids. 
Using the new \emph{contraction representations} of Koana and the author~\cite{KoanaW25stacs},
this gap was filled in Section~\ref{ssec:rep-new} to handle general
linear delta-matroids, as well as extending the result to
\emph{projected} linear delta-matroids. 

Second, given the context of delta-matroids, it may be natural to ask 
for a variant of the notion of representative set, as follows.
Let $D=(V,\cF)$ be a delta-matroid and $X, Y \subseteq V$.
Say that $Y$ \emph{delta-extends} $X$ if $X \Delta Y \in \cF$.
Given that the symmetric difference operation $X \Delta Y$
is fundamental for delta-matroids, it is natural to ask for a
representative set with respect to this notion, instead of the current
requirement that $X \cap Y = \emptyset$. As stated in~\cite{Wahlstrom24repset},
we are not aware of any applications of such a result; we merely note
that it is a natural question which does not appear to immediately
follow from current tools.

We can also consider the following. Let $D=(V,\cF)$ be a delta-matroid
and define the \emph{diameter of $D$} as
$\max_{S, T \in \cF} |S \Delta T|$. We note that if $D$ is linear,
then the diameter is related to the minimum rank of the matrix in a
twist representation of $D$. The results of this paper are stated in
terms of the rank of the terminal set $T$. Can this be extended to a
setting where the bounds (and/or definitions) are adapted to be stated
in terms of the diameter parameter instead? The challenge here is not
just to come up with a valid variation of Theorem~\ref{thm:dm-extend-rank}
of Corollary~\ref{cor:dm-extend-full}, 
but to identify one which carries additional utility.

Finally, there is a gap between the upper and lower bounds. Given a
directly represented delta-matroid $D=(V,\cF)$ and a terminal set
$T \subseteq V$ of rank $k$ in $D$, a representative set $\rep{S}$ for
a collection $S$ of $q$-sets in $(D,T)$ computed in this paper has
cardinality roughly $|\rep{S}|=O((qk)^q)$, whereas the lower bounds only
require $|\rep{S}| \geq \binom{k}{q}$. Can the bound on $|\rep{S}|$
be improved, closer to the bound $\binom{k+q}{q}$ for linear matroids?
However, we note that the determinantal sieving method and its
extension to Pfaffians~\cite{theoretics:14026,abs-2502-13654}
alleviate some of the need for such a result. 

\paragraph{Questions: Kernelization.}
The results raise multiple questions for kernelization. Let us point
out a few. There are a few questions in kernelization which
have connections to non-bipartite graph matching, where delta-matroid
methods might help. \textsc{Edge Dominating Set}, equivalent to the
problem of computing a minimum maximal matching, has a kernel of
$O(k^2)$ vertices and $O(k^3)$ edges~\cite{XiaoKP13eds}, and it is a
long-open question whether either of these bounds can be
improved~\cite{Book_kernelization_FLSZ19}. Another, perhaps more obscure
question is \textsc{Directed Multiway Cut with Deletable Terminals} (DTDMWC).  
In the standard formulation, \textsc{Directed Multiway Cut}
does not allow a polynomial kernel unless the polynomial hierarchy
collapses, even for the case of 2 terminals~\cite{CyganKPPW14nopk}.
However, the variant where the solution is allowed to delete terminals 
appears easier. In fact, the dual path-packing problem for DTDMWC can
be solved in polynomial time and has a min-max characterization via a
reduction to matching in an auxiliary graph. Does DTDMWC have a
(presumably randomised) polynomial kernel? This would be analogous to
the undirected version, where a kernel for \textsc{Multiway Cut} in
general is open but \textsc{Multiway Cut with Deletable Terminals} 
is known to have a kernel of $O(k^3)$ vertices~\cite{KratschW20JACM}.

Finally, it is obvious to ask whether the tools of Mader
delta-matroids and Mader-mimicking networks can settle any further
cases of kernelization for \textsc{Multiway Cut}. So far, this is only
known for planar graphs~\cite{JansenPL21planarpk} and for a constant
number of terminals~\cite{KratschW20JACM}. As an intermediate result,
does \textsc{Vertex Multiway Cut} have a quasipolynomial kernel,
as is known for \textsc{Edge Multiway Cut}~\cite{Wahlstrom22talg}?\footnote{There
is preliminary work in this direction, but not conclusive yet.}

Naturally, there is also any number of other open questions in
kernelization that are less immediately related to the tools developed
in this paper; see~\cite{Book_kernelization_FLSZ19} for a selection.

\paragraph{Matching-mimicking networks.} 
Let $(G,T)$ be a terminal network. A \emph{matching-mimicking network}
for $(G,T)$ is a terminal network $(G',T)$ such that
for every $S \subseteq T$, there is a perfect matching in $G-S$ if and
only if there is a perfect matching in $G'-S$. 
In particular, if $G$ is bipartite then constructing a
matching-mimicking network reduces back to graph cuts,
hence a matching-mimicking network of size $O(|T|^3$)
for bipartite graphs follows from the cut-covering sets lemma.
However, the case of non-bipartite graphs appears (unsurprisingly)
harder. Eppstein and Vazirani~\cite{EppsteinV21} defined
the notion and noted that a matching-mimicking network with \emph{some}
size bound $|V(G')| \leq f(|T|)$ always exists, simply since $(G,T)$
for finite $|T|$ has a finite number of behaviours. However, they gave
no explicit bound and no method of constructing them.
See also follow-up blog posts by Eppstein~\cite{Eppblog1,Eppblog2}.

We note that an algebraic representation as a matrix encoded into
$O(|T|^3\log n)$ bits can be constructed via a contraction of the
matching delta-matroid (and in fact, one with $O(|T|^3)$ bits exists
via a less trivial argument). 
However, this does not imply anything for the combinatorial question
of the minimum size of a matching-mimicking network.

Is there always a matching-mimicking network $(G',T)$ where
$|V(G')|=|T|^{O(1)}$? Furthermore, to repeat another question of
Epstein and Vazirani~\cite{EppsteinV21}, are there matching-mimicking
networks of bounded size for \emph{weighted} matchings?

\paragraph{Path-packing systems.}
Multiple extensions and variations on Mader's path-packing theorem
have been considered in the literature, including packing non-zero
$T$-paths in a group-labelled graph~\cite{ChudnovskyGGGLS06,ChudnovskyCG08}
and packing ``non-returning'' $T$-paths~\cite{Pap07,Pap08}.
Does the existence of a linear Mader delta-matroid extend suitably to
these settings? 

Furthermore, for the setting of group-labelled graphs,
Chudnovsky et al.\ show the existence of a related matroid~\cite{ChudnovskyGGGLS06}
and an object similar in structure to a delta-matroid (although not
precisely, and not by name)~\cite{ChudnovskyCG08},
which capture more of the path-packing structure
than the Mader matroids and delta-matroids do. 
Can these be suitably captured by linear representations, and constructed in polynomial time?

\paragraph{Acknowledgements.}
Part of the work was done at the 2019 IBS summer research program on
algorithms and complexity in discrete structures, Daejeon, South Korea.
We are grateful to Eun Jung Kim and Stefan Kratsch for discussions on
topics related to the paper.

\bibliographystyle{abbrv}
\bibliography{bib}

\end{document}